\documentclass[final]{elsarticle}
\usepackage[disable]{todonotes} 
\newcommand{\itamar}[1]{}
\newcommand{\gabi}[1]{}
\journal{Computer Networks}

\usepackage{lineno,hyperref}
\modulolinenumbers[5]

\usepackage{pgfplots}

\usepackage{amsmath}
\usepackage{amsthm}
\usepackage{tabto}
\usepackage{algpseudocode}
\usepackage{algorithm}
\usepackage{xspace}
\algnewcommand{\AND}{\textbf{and}\xspace}
\algnewcommand{\OR}{\textbf{or}\xspace}
\usepackage{paralist}
\PassOptionsToPackage{hyphens}{url}
\usepackage{float}
\usepackage{graphicx}
\usepackage{subcaption}
\usepackage{color}
\makeatletter
\def\blfootnote{\xdef\@thefnmark{}\@footnotetext}
\makeatother

\makeatletter
\newcommand*{\algrule}[1][\algorithmicindent]{\makebox[#1][l]{\hspace*{.5em}\vrule height .75\baselineskip depth .25\baselineskip}}%

\newcount\ALG@printindent@tempcnta
\def\ALG@printindent{%
    \ifnum \theALG@nested>0
        \ifx\ALG@text\ALG@x@notext
            \addvspace{-3pt}
        \else
            \unskip
            \ALG@printindent@tempcnta=1
            \loop
                \algrule[\csname ALG@ind@\the\ALG@printindent@tempcnta\endcsname]%
                \advance \ALG@printindent@tempcnta 1
            \ifnum \ALG@printindent@tempcnta<\numexpr\theALG@nested+1\relax
            \repeat
        \fi
    \fi
    }%
\usepackage{etoolbox}
\patchcmd{\ALG@doentity}{\noindent\hskip\ALG@tlm}{\ALG@printindent}{}{\errmessage{failed to patch}}
\makeatother

\newtheorem{theorem}{Theorem}
\newtheorem{corollary}[theorem]{Corollary}
\newtheorem{lemma}[theorem]{Lemma}
\newtheorem{proposition}[theorem]{Proposition}

\newcommand{\set}[1]{\left\{#1\right\}}

\newcommand{\revchange}[1]{{\color{black} #1}}
\newcommand{\revchangetwo}[1]{{\color{black} #1}}

\DeclareMathOperator{\mSAM}{SA} 
\newcommand{\SAM}{$\mSAM$}

\DeclareMathOperator{\mSAMWP}{SA} 
\newcommand{\SAMWP}{$\mSAMWP$}

\DeclareMathOperator{\mSAO}{SA^*} 
\newcommand{\SAO}{$\mSAO$}
\DeclareMathOperator{\mSAOWP}{SA^*} 
\newcommand{\SAOWP}{$\mSAOWP$}

\newcommand{\alg}{Alg}
\DeclareMathOperator{\malg}{Alg} 
\newcommand{\opt}{OPT}
\newcommand{\makeroom}{MakeRoom()}

\newcommand{\subopt}{SubOPT}
\DeclareMathOperator{\msubopt}{SubOPT} 
\DeclareMathOperator{\FILL}{fill}
\DeclareMathOperator{\FLUSH}{flush}

\DeclareMathOperator{\W}{CW} 
\newcommand{\Cs}{G} 
\newcommand{\CsK}{G^K} 
\newcommand{\CsU}{G^U} 

\newcommand{\mFILL}{^{(\FILL)}}
\newcommand{\mFLUSH}{^{(\FLUSH)}}
\renewcommand{\W}{^{(W)}}
\renewcommand{\P}{^{(P)}}
\newcommand{\U}{^{(U)}}
\newcommand{\K}{^{(K)}}

\newcommand{\hfull}{Gfull}
\newcommand{\Autp} {A^{\U}_{(t_p)}}
\newcommand{\Au}{A^{\U}} 
\newcommand{\Ak}{A^{\K}} 
\newcommand{\malpha}{\alpha} 
\newcommand{\deltaW}{\delta_W} 
\newcommand{\deltaP}{\delta_V} 
\newcommand{\mW}{\ell_W} 
\newcommand{\mP}{\ell_V} 
\newcommand{\minWork}{W_0} 
\newcommand{\mbest}{(\minWork,V)} 
\newcommand{\mworst}{(W,1)} 
\newcommand{\algc}{\malg_c} 

\begin{document}

\begin{frontmatter}

\title{Queueing in the Mist:\\Buffering and Scheduling with Limited Knowledge}
\author{Itamar Cohen and Gabriel Scalosub}

\address{   Department of Communication Systems Engineering,
   Ben-Gurion University of the Negev,
   Beer-Sheva 84105, Israel\\
   Email:
   {\tt itamarq@post.bgu.ac.il, sgabriel@bgu.ac.il}
}

\begin{abstract}
Scheduling and managing queues with bounded buffers are among the most fundamental problems in computer
networking. Traditionally, it is often assumed that all the properties of each packet are known immediately upon arrival. However, as traffic becomes increasingly heterogeneous and complex, such assumptions are in many cases invalid. In particular, in various scenarios information about packet characteristics becomes available only after the packet has undergone some initial processing.
In this work, we study the problem of managing queues with limited knowledge. We start by showing lower bounds on the competitive ratio of any algorithm in such settings.
The techniques used in our proofs, which make use of a carefully crafted Markov process, may be of independent interest, and can potentially be used in other similar settings as well.
Next, we use the insight obtained from these bounds to identify several algorithmic concepts appropriate for the problem, and use these guidelines to design a concrete algorithmic framework. We analyze the performance of our proposed algorithm, and further show how it can be implemented in various settings, which differ by the type and nature of the unknown information. We further validate our results and algorithmic approach by an extensive simulation study that provides further insights as to our algorithmic design principles in face of limited knowledge.
\end{abstract}


\end{frontmatter}


\section{Introduction}
\label{sec:introduction}
\blfootnote{An earlier version of this work was published in~\cite{Qmist_IWQoS}. This work adds full proofs of all theorems, stronger lower bounds, an improved competitive algorithm, and an extended simulation study.}

Some of the most basic tasks in computer networks involve scheduling and managing queues equipped with finite
buffers, where the primary goal in such settings is maximizing the throughput of the system.
The always-increasing heterogeneity and complexity of network traffic makes the challenge of maximizing the
throughput ever harder,
as the packet processing required in such queues
spans a plethora of tasks including
various forms of DPI, MPLS and VLAN tagging, encryption / decryption, compression / decompression, and more.

The most prevalent assumption in the research studying these problems is that the various properties of any packet -- e.g., its QoS characteristic, its required processing, its deadline -- are known upon its
arrival.
However, this assumption is in many cases unrealistic.
For instance, when a packet is recursively encapsulated a few times by MPLS, PBB, 802.1Q, GRE or IPSec, it is hard
to determine in advance the total number of processing cycles that such a packet would require~\cite{Folded,
Kangaroo}.
Furthermore, the QoS features of a packet are commonly determined by its flow ID, which is in many cases known only after parsing~\cite{Kangaroo}.

In data center network architectures such as PortLand~\cite{PortLand}, ingress switches query a cache for an
application-to-location address resolution. A cache miss, which is unpredictable by nature, results in forwarding of
the packet to the switch software or to a central controller, which performs a few additional processing cycles
before the packet can be transmitted.
Similarly, in the realm of Software Defined Networks, ingress switches query a cache for obtaining rules for a
packet~\cite{DIFANE}, which may also depend on priorities~\cite{ETHANE}. In such a case, a cache miss results in
additional processing until the rules are retrieved and the profit from the packet is known.

In spite of this increased heterogeneity, and the fact that the processing requirement of a packet might not be known in advance, these characteristics usually become known once some initial processing is performed. This behavior is common in many of the applications just described. Furthermore, for traffic corresponding to the same flow, it is common for characteristics to be unknown when the first few packets of the flow arrive at a network element, and once these properties are unraveled, they become known for all subsequent packets of this flow.
\revchange{It therefore follows that only part of the arriving packets has unknown characteristics upon arrival, which become known after parsing.}

In this work we address such scenarios where the characteristics of some arriving traffic are unknown upon arrival,
and are only revealed when a packet has undergone some initial processing
(parsing), ``causing the mist to clear''.
We model and analyze the performance of algorithms in such settings, and in particular we develop online scheduling and buffer management algorithms for the problem of maximizing the profit obtained from delivered packets, and
provide guarantees on their expected performance using competitive analysis.

We focus on the general case of heterogeneous processing requirements (work) and heterogeneous
profits~\cite{MultiV_MultiW}.
We assume priority queueing, where the exact priorities depend on the specifics of the model studied.
We present both algorithms and lower bounds for the problem of dealing with unknown characteristics in these models.
Furthermore, we highlight some design concepts for settings where algorithms have limited knowledge, which we
believe might be applicable to additional scenarios as well.

\begin{figure}[t]
\centering
\includegraphics[width=0.8\columnwidth]{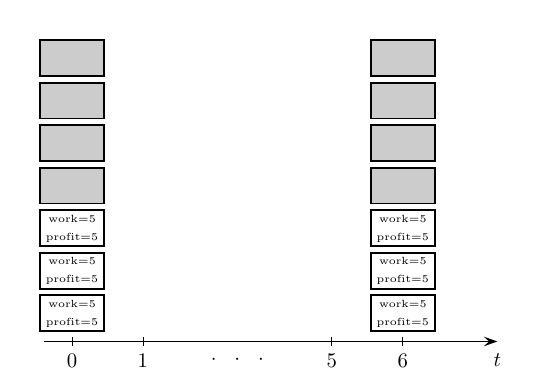}
\caption{An illustrative example of an arrival sequence with known and unknown packets}
\label{fig:toy_example}
\end{figure}

As an illustration of the problem, assume we have a 3-slots buffer, equipped with a single processor, and consider
the arrival sequence depicted in Fig.~\ref{fig:toy_example}.
In the first cycle we have seven unit-size packets arriving, out of which three will provide a profit of 5 upon successful delivery, each requiring 5 processing cycles (work).
The characteristics of these three packets are known immediately upon arrival.
\revchangetwo{The} characteristics of the remaining four packets (marked gray) are \emph{unknown} upon arrival. We therefore dub
such packets $U$-packets (i.e., unknown packets). Each of these four $U$-packets may turn out to be either a "best"
packet, requiring minimal work and having maximal profit; a "worst" packet, requiring maximal work and having
minimal profit; or anything in between. Thus, already at the very beginning of this simple scenario, any buffering
algorithm would encounter an \emph{admission control} dilemma: how many $U$-packets to accept, if any?
This dilemma can be addressed by various approaches including, e.g., allocating some buffer space for $U$-packets,
accepting $U$-packets only when current known packets in the buffer are of poor characteristics, in terms of profit,
or of profit to work ratio, etc. In case that the algorithm accepts $U$-packets, an additional question arises:
which of the $U$-packets to accept into the buffer? Obviously, for any online deterministic algorithm there exists a
simple adversarial scenario, which would cause it to accept only the "worst" $U$-packets (namely, packets with
maximal work and minimum profit), while an optimal offline algorithm would accept the best packets. This motivates
our decision to focus our attention on randomized algorithms.

We now turn to consider another aspect of handling traffic with some unknown characteristics. Assume the scenario continues with 5 cycles without any arrival, and then a cycle with an identical arrival pattern - namely, three known packets with both work and profit of 5 per packet, and four $U$-packets. This
sheds light on a \emph{scheduling} dilemma: which of the accepted packets should better be processed first? every
scheduling policy impacts the buffer space available in the next burst. For instance, a run-to-completion attitude
would enable finishing the processing of one known packet by the next burst, thus allowing space for accepting a new
packet without preemption. However, one may consider an opposite attitude - namely, parsing as many $U$-packets as
possible, thus "causing the mist to clear",
allowing more educated decisions, once there are new arrivals. In terms
of priority queuing, this means over-prioritizing some
 $U$-packets, and allowing them to be parsed immediately upon arrival.
We further develop appropriate algorithmic concepts based on the insights from this illustrative example in Section \ref{sec:algorithmic_conecpts}.

\subsection{System Model}
\label{sec:model}
Our system model consists of four main modules, namely, \begin{inparaenum}[(a)]
\item an input queue equipped with a finite buffer,
\item a buffer management module which performs admission control
\item a scheduler module which decides which of the pending packets should be processed, and
\item a processing element (PE), which performs the processing of a packet.
\end{inparaenum}

We divide time into discrete cycles,
\revchange{where each cycle represents a fixed time slot,}
 and consists of three steps:
\begin{inparaenum}[(i)]
\item The {\em transmission} step, in which fully-processed packets leave the queue,
\item the {\em arrival} step, in which new packets may arrive, and the buffer management module decides which of them should be retained in the queue, and which of the currently buffered packets should be pushed-out and dropped, and finally
\item the {\em processing} step, in which the scheduler assigns a single packet for processing by the PE, which in turn processes the packet.
\end{inparaenum}

We consider a sequence of unit-size packets arriving at the queue. Upon its arrival, the characteristic of each packet may be \emph{known} - in which case we refer to the packet as a {\em $K$-packet} (i.e., known packet); or \emph{unknown}  - in which case we refer to the packet as a {\em $U$-packet} (i.e., unknown packets). We let $M$ denote the maximum number of $U$-packets that may arrive in any single cycle.
We focus our attention on the case where $M>0$, unless specifically stated otherwise.

Each arriving packet $p$ has some
\revchangetwo{
\begin{inparaenum} [(1)]
\item intrinsic benefit (\emph {profit})
$v(p) \in \set{1,\dots,V}$, and
\item required number of processing cycles (\emph {work}),
$w(p) \in \set{\minWork, \minWork+1, \dots ,W}$. Unless explicitly stated otherwise, we consider the most general case, namely, $\minWork = 1$.
\end{inparaenum}
}
To simplify the expressions throughout the paper, we assume that both $V$ and $W$ are powers of 2.\footnote{Our results degrade by a mere constant factor otherwise.}
We use the notation {\em $(w,v)$-packet} to denote a packet with work $w$ and profit $v$.
We note that the {\em uniform} case where all packets require the same amount of work, and all packets have the same profit, is trivial, since the simple run-to-completion policy is optimal. We therefore focus our attention on non-uniform traffic.

In our model, similarly to~\cite{Shpiner}, upon processing a $U$-packet for the first time, its properties become known. We therefore refer to such a first processing cycle of a $U$-packet as a {\em parsing cycle}. Non-parsing cycles where the processor is not idle are referred to as {\em work cycles}.

The queue buffer can contain at most $B$ packets.
We assume $B \geq 2$, since the case where $B = 1$ is degenerate. The {\em head-of-line} (HoL) packet at time $t$
(for a given algorithm \alg) is the highest priority packet stored in the buffer just prior to the processing step
of cycle $t$, namely, the packet to be scheduled for processing in the processing step of $t$.
We say the buffer is {\em empty} at cycle $t$ if there are no packets in the buffer after the transmission step of cycle $t$.

We study {\em queue management} algorithms, which are responsible for both the buffer management
and the scheduling of packets for processing. In particular, we focus our attention on algorithms targeted at
maximizing the {\em throughput} of the queue, i.e. the overall profit from all packets successfully transmitted out
of the queue.
\revchange{The throughput of algorithm \alg\ is denoted by $TP(\malg)$. We use the terms throughput and performance interchangeably.}

We evaluate the performance of online algorithms using competitive analysis~\cite{Amortized, El-Yaniv}. An algorithm
\alg\ is said to be $c$-competitive if for every finite input sequence $\sigma$, the throughput of {\em any}
algorithm for this sequence is at most $c$ times the throughput of \alg\ ($c \geq 1$).
We let \opt\ denote any (possibly clairvoyant) algorithm attaining optimal throughput.
An algorithm is said to be {\em greedy} if it accepts packets as long as there is available buffer space. We further
focus our attention on {\em work-conserving}
algorithms, i.e., algorithms which never leave the PE idle unnecessarily.

\subsection{Related Work}
Competitive algorithms for scheduling and management of bounded buffers have been extensively studied for the past two decades. The problem was first introduced in the context of differentiated services, where packets have uniform size and processing requirements, but some of the packets have higher priorities,
represented by a higher profit associated with them~\cite{Rosen, OF, Smoothing}. The numerous variants of this problem include models where packets have deadlines or maximum lifetime in the switch~\cite{OF}, environments involving multi-queues~\cite{Albers2005, Azar&Richter2004, Buf_Xbar_Kogan, Buf_Xbar_Kanizo} and cases with packets dependencies~\cite{Dependencies, OF_sizes}, to name but a few.
An extensive survey of these models and their analysis can be found in~\cite{Goldwasser}.

While traditionally it was assumed that packets have heterogeneous profits but uniform work (processing
requirements), some recent work introduced the complementary problem, of uniform profits with heterogeneous
work~\cite{Multipass}. This work presented an optimal algorithm for the fundamental problem, as well as
online algorithms and bounds on the competitive ratio for numerous variants. Subsequent research investigated related
problems with heterogeneous work combined with heterogeneous packet sizes~\cite{Sizes}, or with heterogeneous
profits~\cite{MultiV_MultiW, Azar2015}.
In particular,~\cite{MultiV_MultiW} showed that the competitive ratio of some straight-forward deterministic algorithms for the problem of heterogeneous work combined with heterogeneous profits is linear in either the maximal work $W$, or in the maximal profit $V$, even when the characteristics of all packets are known upon arrival. These results motivate our focus on randomized algorithms.

While most of the literature above assumed that all the characteristics of packets are known upon arrival, this
assumption was put in question recently~\cite{Shpiner} by noting that it is often invalid.
However, the main problem addressed in~\cite{Shpiner} revolved around developing schemes for transmitting packets of the same flow in-order, while our work focuses on maximizing throughput with limited buffering resources, and designing both buffer management and scheduling policies targeted at this objective.

Maybe closest to our work are the recent studies considering serving in the dark~\cite{Dark,Dark2}, which investigate an extreme case where the online algorithm learns the profit from a packet only after transmitting it. These studies
consider highly oblivious algorithms, whereas our model and our proposed algorithms dwell in a middle-ground between
the well studied models with complete information, and these recent oblivious settings. Our work further considers
traffic with variable processing requirements, whereas~\cite{Dark,Dark2} focus on settings where all packets require
only a single processing cycle, and they differ only by their profit.

The problem of optimal buffering of packets with variable work is closely related to the problem of job scheduling
in a multi-threaded processor, which was extensively studied. A comprehensive survey of online algorithms for this
problem can be found in~\cite{Job}. This body of work, however, differs significantly from our currently studied
model.
The major differences are that packet buffering has to deal with limited buffering capabilities, and is targeted at
maximizing throughput. Processor job scheduling, however, usually has no strict buffering limitations, and is mostly
concerned with minimizing the response time.

\subsection{Our Contribution}
We introduce the problem of buffering and scheduling which aims to maximize throughput where the characteristics of
some of the packets are unknown upon arrival. We focus our attention on traffic where every packet has some required
processing cycles, and some profit associated with successfully transmitting it.
\revchange{We make no assumption on the underlying process generating traffic, thus rendering our results globally applicable.}

In Section~\ref{sec:lower_bounds} we present lower bounds on the performance of any randomized algorithm for the
problem. Specifically, we show that no algorithm can have a competitive ratio better than $\Omega(\min\set{WV,M})$, even against an adversary which can accommodate merely 2 packets in its buffer,
\revchangetwo{where $W$ and $V$ denote the maximum work and profit of a packet, respectively, and $M$ represents the maximum number of unknown packets which may arrive in any single cycle.}
We also prove stronger lower bounds for the general settings using a novel technique, in which we bound the expected number of packets in the buffer of an optimal offline algorithm by means of a Markov process.

In Section~\ref{sec:algorithmic_conecpts} We describe several algorithmic concepts tailored for dealing with unknown characteristics in such systems.
We follow by presenting an algorithm that applies our suggested algorithmic concepts in Section~\ref{sec:algorithms}.
For the most general case we prove our algorithm has a competitive ratio of $O(M \log V \log W)$. We further show how to improve this bound in several important special cases.

In Sections~\ref{sec:improved_algorithms}-\ref{sec:Practical_Implementation}
we present some modifications and heuristics applicable to our algorithm that, while leaving the worst-case guarantees intact, are designed to improve performance compared to the baseline algorithmic design.
The modified algorithm can cope with cases where neither the maximal amount of work and profit, nor the maximum number of unknown packets per cycle, are known in advance.

We further validate and evaluate the performance of our proposed algorithms in Section~\ref{sec:simulations} via an
extensive simulation study. Our results highlight the effect the various parameters have on the problem, well beyond
the insights arising from our rigorous mathematical analysis.

We conclude in Section~\ref{sec:conclusions} with a discussion of our results, and also highlight several
interesting open questions.

\section{Lower Bounds}
\label{sec:lower_bounds}
In this section we present lower bounds on the competitive ratio of any randomized algorithm for our problem.
\revchange{
These lower bounds serve two main objectives:
\begin{inparaenum}[(i)]
\item They represent the best competitive ratio which one can hope to achieve; and
\item the hard scenarios used in the proofs of these lower bounds highlight the challenges which any competitive online algorithm would have to tackle.
\end{inparaenum}
}

\subsection{Highly-restricted adversaries}\label{sec:lower_bounds_restricted}
\revchange{
In this section we prove lower bounds on the competitive ratio of any online algorithm for our problem, compared to a highly-restricted adversary which uses a buffer which can only store a single packet. This restriction on the amount of buffer space available for the adversary enables us to better highlight the scaling laws of the problem, depending on the various parameters.
}

\begin{theorem} \label{Theorem_WP_rand}
If $V \geq 1$, $M \geq 1$ and the work of each packet is $w(p) \in \set{\minWork, \minWork+1, \dots, W}$ where $W \geq 2$,
then the competitive ratio of any randomized algorithm  for non-uniform traffic is at least
$$
\frac{V (W-1)}{2\minWork} \left[1 - \left(1- \frac{1} {V(W-1)+1-\minWork}\right)^{M\minWork} \right],
$$

even against an optimal offline algorithm which has a buffer which can only store a single packet.
\end{theorem}

\begin{proof}
Since traffic is non uniform, we are guaranteed to have $V(W-1) + 1 - \minWork \neq 0$.
\revchange{
We prove the theorem using Yao's method~\cite{yao77probabilistic}, where we define a carefully crafted distribution over arrival sequences, and show a lower bound on the ratio between the expected performance of an optimal clairvoyant algorithm for the problem, and the expected performance of any {\em deterministic} algorithm for the problem.
}
We will show that the claim is true even if the optimal offline algorithm uses a buffer that can hold only a single packet.
We define the following collection of arrival sequences, where each arrival sequence has two phases: a {\em Fill phase}, and a {\em Flush phase}.
The Fill phase consists of \emph{iterations} as follows. Each iteration begins with $\minWork$ cycles without arrivals; and continues with $\minWork$ cycles with $M$ $U$-packets arriving per cycle, where each packet is a $\mbest$-packet with probability $p$, and a $\mworst$-packet with probability $(1-p)$, for some constant $p$ to be determined later. The total number of cycles during the fill phase is $N$, where $N$ is a large integer, so we have $\frac{N}{2 \minWork}$ iterations.
Once the fill phase ends, it is followed by the Flush phase, which consists of $BW$ cycles without arrivals.
We note that due to the random choices of packets being either $\mbest$-packets or $\mworst$-packets, the above structure induces a distribution over a collection of possible arrival sequences.

To simplify our analysis, we define the \subopt\ policy, which works as follows: Within the fill phase, during each iteration, \subopt\ accepts at most one $\mbest$-packet which has arrived during the iteration, if such a packet exists. This packet is the one considered {\em picked} by \subopt\ in that iteration.
Starting from the second iteration, during the first $\minWork$ cycles of each iteration, \subopt\ processes the packet it picked during the previous iteration (if such a packet exists), and transmits it.
During the flush phase, \subopt\ processes and finally transmits the packet it picked during the last iteration.

It should be noted that \subopt\ is neither greedy, nor work conserving.
Moreover, the expected throughput of \subopt\ clearly serves as a lower bound on the expected optimal throughput
possible.

We have $\frac{N}{2\minWork}$ iterations, and the probability that \subopt\ successfully picks a $\mbest$-packet during an iteration is exactly the probability of there being a $\mbest$-packet arriving during that iteration, which is $1 - (1-p)^{M\minWork}$.
\revchange{The throughput of \subopt, which we recall is denoted by $TP(\msubopt)$, therefore satisfies}
\begin{equation} \label{TP_Opt_rand_WP}
TP(\msubopt) \geq \frac{NV}{2\minWork}[1 - (1-p)^{M\minWork}]
\end{equation}

We now turn to consider the expected performance of any deterministic algorithm \alg\ for the problem.
We first assume that \alg\ begins the flush phase with a buffer
full of $\mbest$-packets,
all of them unparsed. This provides \alg\ with a profit of $BV$ during the flush phase, while
still having $N$ processing cycles during the fill phase for processing additional packets. This
profit is clearly an upper bound on the maximum possible throughput attainable by \alg\ from packets transmitted during the flush phase, regardless of when they were processed.
For evaluating the gain of \alg\ during the fill phase, it therefore suffices to consider only packets which \alg\ fully processes during this phase.

Consider now the profit of \alg\ from packets transmitted during the fill phase.
Recall that we assume that \alg\ is work-conserving. We assume that \alg\ is also greedy, that is, \alg\ never discards a packet when its buffer is not full; being greedy cannot decrease \alg's performance. \alg\ has packets to process during the entire fill phase, except for the first $\minWork$ cycles (where there are no arrivals yet), namely, for $N' = N - \minWork$ cycles.
Furthermore, since \alg\ is assumed to always accept packets when the buffer is not full, and is work conserving, there exists some $0 < r \leq 1$
such that the number of parsing, and work, cycles performed by \alg\ are $N'r$, and $N'(1-r)$, respectively.

Consider a case where \alg\ reveals a $\mbest$-packet $q$. Then, processing $q$ and finally transmitting it would surely
not decrease the throughput of \alg\ when contrasted with the alternative of dropping $q$.
Thus, the best deterministic algorithm \alg\ would work at least $\minWork-1$ work cycles per each parsing cycle, in which
a $\mbest$-packet is parsed
(recall that we are merely interested in packets, which \alg\ fully processes and transmits during the fill phase).
Therefore, the total number of work cycles contributing to the transmission of such packets is at least $\minWork-1$ times larger then the expected number of parsing
cycles, in which a $\mbest$-packet is revealed: $N'(1-r) \geq N'rp(\minWork-1)$.

If the total number of work cycles during the fill phase exceeds the number of cycles which are necessary for
transmitting all the parsed $\mbest$-packets, \alg\ may work also on $\mworst$-packets. Namely, if $N'(1-r) > N'rp(\minWork-1)$,
then \alg\ may work on $\mworst$-packets for
$N'(1-r) - N'rp(\minWork-1)$ cycles, transmitting at most one $\mworst$-packet once in $W-1$ such cycles.

Combining the above reasoning we conclude that the overall throughput of \alg\ satisfies
\begin{equation}\label{TP_alg_lower_b_WP}
\begin{split}
TP(\malg) & \leq N'rpV + \frac {N'(1-r) - N'rp(\minWork-1)} {W-1} + BV \\
& = (N - \minWork) \left[Vrp + \frac {(1-r) - rp(\minWork-1)} {W-1}\right] + BV
\end{split}
\end{equation}

Considering the ratio between the lower bound on the expected performance of \subopt\ (as captured by Eq.~\ref{TP_Opt_rand_WP}) and the upper bound on the expected performance of \alg\ (as captured by Eq.~\ref{TP_alg_lower_b_WP}) and letting $N \to \infty$, we conclude that no algorithm can have a competitive ratio better than
$$
\frac{V(W-1)}{2\minWork} \cdot \frac{1 - (1-p)^{M\minWork}}{Vrp(W-1)+ 1-r -rp(\minWork-1)}
$$
By choosing $p^* = \left[ V(W-1)+1-\minWork \right]^{-1}$, the result follows.
\end{proof}

We now aim to relate the lower bound established in Theorem~\ref{Theorem_WP_rand} to a simpler and more intuitive function of $M, V$ and $W$. We do so by means of two propositions, which relate the bound to either $\Omega(M)$ or $\Omega(VW)$ for different ranges of $M$. In the propositions we use our notation $p^* = \left[ V(W-1)+1-\minWork \right]^{-1}$ from the proof of Theorem~\ref{Theorem_WP_rand}.
\revchange{
Using this notation, note that Theorem~\ref{Theorem_WP_rand} shows that the competitive ratio is at least
$$
\frac{V(W-1)}{2\minWork} \left[1 - \left(1- p^* \right)^{M\minWork} \right].
$$

In the proofs of both propositions we will repeatedly use the following simple inequality, which holds for any $\minWork \geq 1$:\emph{•}
\begin{equation} \label{simple_inequality_p_star}
\frac{1}{V(W-1)} = \frac{1}{\frac{1}{p^*} + \minWork - 1} \leq p^*.
\end{equation}
}

The following proposition shows that if $M$ is relatively small, then the lower bound established in Theorem~\ref{Theorem_WP_rand} is $\Omega (M)$.
\revchange{
\begin{proposition} \label{prop:M_over_2_VW}
If $V \geq 1, \minWork \geq 1, W \geq 2$ and $1 \leq M \leq \frac{V(W-1)}{\minWork}$, then
\begin{equation}
\frac{V(W-1)}{2\minWork} \left[1 - \left(1- p^* \right)^{M\minWork} \right] \geq \frac{M}{4} \notag
\end{equation}
\end{proposition}

\begin{proof}
We show by induction on $n$ that for any $1 \leq n \leq V(W-1)$
\begin{equation}
\label{eq:induction_claim}
(1-p^*)^n \leq 1 - \frac{n}{2V(W-1)}.
\end{equation}
By setting $n=M \cdot \minWork$, which is at most $V(W-1)$ by our assumption on $M$, and applying some algebraic manipulation, the result follows.

For $n=1$, Eq.~\ref{eq:induction_claim} reduces to requiring that $\frac{1}{2V(W-1)} \leq p^*$, which holds true due to Eq.~\ref{simple_inequality_p_star}.
For the induction step, by the
induction hypothesis on $n$ we have
$$
\left(1-p^*\right)^{n+1}
\leq
\left(1 - p^*\right) \left[  1 - \frac{n}{2 V(W-1)} \right].
$$
It therefore suffices to prove that
$$
\left(1 - p^*\right) \left[  1 - \frac{n}{2 V(W-1)} \right]
\leq
1 - \frac{n+1}{2 V(W-1)},
$$
which is equivalent to requiring that
$$
\frac {1} {2 V(W-1)}
\leq
p^* \left[ 1 - \frac{n}{2 V(W-1)} \right].
$$
By Eq.~\ref{simple_inequality_p_star} we have
$\frac{1}{2 V(W-1)} \leq \frac{p^*}{2}$, which implies that it suffices to show that
$$
\frac{p^*}{2}
\leq
p^* \left[ 1 - \frac{n}{2 V(W-1)} \right]
$$
which is satisfied for every $n \leq V(W-1)$.
\end{proof}
}

The following proposition shows that if $M$ is relatively large, then the lower bound established in Theorem~\ref{Theorem_WP_rand} is $\Omega \left(\frac{VW}{\minWork}\right)$.

\revchange{
\begin{proposition}\label{prop:omegaVW}
If $V \geq 1, \minWork \geq 1, W \geq 2$ and $M > \frac{V(W-1)}{\minWork}$, then
\begin{equation}
\frac{V(W-1)}{2\minWork} \left[1 - \left(1- p^* \right)^{M\minWork} \right]
>
\frac{e-1}{4e} \cdot \frac{VW}{\minWork} \notag
\end{equation}
\end{proposition}

\begin{proof}
By our assumption on $M$, and using Eq.~\ref{simple_inequality_p_star}, we have $M \cdot \minWork > V(W-1) \geq \frac{1}{p^*}$.
It follows that
$M \cdot \minWork = a \frac{1}{p^*}$ for some $a >1$, which in turn implies that
$$
\left(1 - p^* \right)^{M \cdot \minWork} = \left[ \left(1 - p^*\right)^{\frac{1}{p^*}}\right]^a \leq e^{-a} <  e^{-1}.
$$
It follows that
\begin{align*}
\frac{V(W-1)}{2\minWork} \left[1 - \left(1- p^* \right)^{M \cdot \minWork} \right]
&\geq \frac{V(W-1)}{2\minWork}  \left(1 - \frac{1}{e}\right) \\ \notag
&= \frac{VW}{2\minWork} \cdot \frac{W-1}{W} \left(1 - \frac{1}{e}\right) \\ \notag
&\geq \frac{e-1}{4e} \frac{VW}{\minWork} \notag
\end{align*}
\end{proof}
}

Assigning $\minWork=1$ in Theorem \ref{Theorem_WP_rand} and Propositions \ref{prop:M_over_2_VW} and \ref{prop:omegaVW} implies the following corollary:
\begin{corollary}\label{cor:WP_ran}
The competitive ratio of any randomized algorithm is $\Omega(\min\set{VW,M})$.
\end{corollary}

In the special case of uniform-profits, we are essentially interested in maximizing the overall {\em number} of packets successfully transmitted. Therefore we may assign $V=1$ in Corollary \ref{cor:WP_ran}, implying the following corollary:

\begin{corollary}\label{cor:W_rand}
In the case of uniform-profits, the competitive ratio of any randomized algorithm is $\Omega(\min\set{W,M})$.
\end{corollary}

In the special case of uniform-work, we can assign $\minWork = W$ in Propositions \ref{prop:M_over_2_VW} and \ref{prop:omegaVW}, implying the following corollary:
\begin{corollary}\label{cor:P_rand}
In the case of uniform-work, the competitive ratio of any randomized algorithm is $\Omega(\min\set{V,M})$.
\end{corollary}

\subsection{Non-restricted adversaries}\label{sec:Markov}
\revchange{In Section~\ref{sec:lower_bounds_restricted} we assumed that the optimal algorithm has a buffer capacity of storing only one packet. This assumption significantly simplified the proofs there.
In this section we relax this assumption, and show a stronger bound for the general, and more natural case, where the size of the buffer available to the optimal algorithm is identical to the size that available to the online algorithm. We use again Yao's method~\cite{yao77probabilistic}, which we used in the proof of Theorem~\ref{Theorem_WP_rand}. Furthermore, we use the same scenario and algorithm \subopt, defined in Section~\ref{sec:lower_bounds_restricted}.
However, as we now allow \subopt\ to store multiple packets in its buffer, \subopt\
can increase its expected throughput by buffering $\mbest$-packets whenever the number of arriving $\mbest$-packets in a single iteration is larger than one, and processing them in iterations where no $\mbest$-packets arrive.
We now evaluate the performance in such settings.}

Denote by $q_j$ the state where there are $j$ $\mbest$-packets in the buffer of \subopt\ at the beginning of an iteration. Note, that when $j > 0$, the count represented by $q_j$ also includes the packet, which is to be transmitted during the iteration. Namely, \subopt\ successfully transmits a packet in every iteration, unless its buffer's state is $q_0$.

We now turn to describe the transition matrix. Denote the probability of having exactly $k$ $\mbest$-packets arriving during one iteration
 by $\malpha_k$. In each iteration we have $M \cdot \minWork$ arriving packets ($M$ packets per cycle, times $\minWork$ cycles per iteration) which are i.i.d. where each packet is a $\mbest$-packet with probability $p$.
Therefore
$\malpha_k = \binom{M \cdot \minWork}{k}p^k (1-p)^{M \cdot \minWork - k}$ when $0 \leq k \leq M \minWork$ and $\malpha_k = 0$ otherwise.

Then, the transition matrix is

\[
\Pi =
\left \{
  \begin{tabular}{ccccccc}
  $ \malpha_0$           & $\malpha_1$ & $\malpha_2$ & \dots       & $\malpha_{B-1}$  & $1 - \sum_{j=0}^{B-1} \malpha_j$ \\
  $ \malpha_0$           & $\malpha_1$ & $\malpha_2$ & \dots       & $\malpha_{B-1}$  & $1 - \sum_{j=0}^{B-1} \malpha_j$ \\
  0                     & $\malpha_0$ & $\malpha_1$ & \dots       & $\malpha_{B-2}$  & $1 - \sum_{j=0}^{B-2} \malpha_j$ \\
  \dots                 & \dots      & \dots      & \dots       & \dots           & \dots \\
  0                     & \dots      & 0          &$\malpha_0$   & $\malpha_1$      & $1 - \sum_{j=0}^{1} \malpha_j$\\
  0                     & 0          & \dots      & 0           & $\malpha_0$      & $1 - \malpha_0$\\
  \end{tabular}
\right \}
\]
where $\Pi_{ij}$ is the probability of transition from state $i$ to state $j$ for each $0 \leq i, j \leq B$.
$\Pi$ is irreducible, because it is possible to get from any buffer state to any other buffer state by some arrival sequence. $\Pi$ is also aperiodic, because its diagonal is non-zero, which represents the fact that if the buffer contains $i$ packets at the beginning of a certain iteration, there exists a positive probability that it would contain $i$ packets also at the beginning of the next iteration.
Furthermore, as $\Pi$ is finite, irreducible and aperiodic, it is also
ergodic, namely, there exists a steady state.
For a long enough input sequence, we can neglect the transient "warm-up" period, and assume that the expected number of iterations where \subopt\ gains nothing during phase 1 is $\frac{N}{2 \minWork} \cdot p_0$, where $p_0$ is the probability that \subopt\ is in state $q_0$. In the rest of the iterations in phase 1 \subopt\ gains $V$ per iteration.
Therefore, the expected throughput of \subopt\ satisfies
\begin{equation} \label{W_rand_SO_B_gt_1}
TP(\msubopt) \geq \frac{N}{2 \minWork} \cdot V(1 - p_0)
\end{equation}

The expected throughput of \alg\  remains the same as in Eq.~\ref{TP_alg_lower_b_WP}.
In order to obtain the competitive ratio for the fully heterogenous case, we divide Eq.~\ref{W_rand_SO_B_gt_1} by Eq.~ \ref{TP_alg_lower_b_WP} and assign again $\minWork=1$ and $p^* = \frac{1}{V(W-1)}$. Then, when $N \to \infty$ the competitive ratio is $c \geq \frac{V}{2}(W-1)(1 - p_0)$.

\revchange{
We find $p_0$ by solving the balance equations defining the steady state of the system, i.e., finding the eigenvector of the transition matrix $\Pi$.}
Fig.~\ref{fig:tight_lower_bnds} depicts the lower bounds as a function of $M$ when $V = W = 10$ for various buffer sizes.
Recall that the probability of a certain packet to be a $\mbest$-packet is $p_* = \frac{1}{V(W-1)} = \frac{1}{90}$. Therefore
only when $M$ is large enough, the expected number of $\mbest$-packets per iteration is sufficient for allowing \subopt\ to really take advantage of its buffer for increasing its performance, resulting in a stronger lower bound on the competitive ratio.

\begin{figure}
\centering
\begin{tikzpicture}
\begin{axis}[
    xlabel = {M},
    ylabel = {Competitive ratio lower bound},
    xmin=0, xmax=97,
    ymin=0, ymax=46,
    xtick={0,20,40,60,80,100},
    ytick={0,5, 10, 15, 20,25, 30, 35, 40, 45},
    legend pos=north west,
    ymajorgrids=true,
    grid style=dashed,
    legend style = {font=\scriptsize},
    label style={font=\scriptsize},
    tick label style={font=\scriptsize},
]

\addplot[color=blue,mark=square,] coordinates {
(1, 0.50000)(6, 2.91789)(11, 5.20441)(16, 7.36669)(21, 9.41148)(26, 11.34517)(31, 13.17379)(36, 14.90306)(41, 16.53836)(46, 18.08481)(51, 19.54724)(56, 20.93021)(61, 22.23803)(66, 23.47479)(71, 24.64436)(76, 25.75037)(81, 26.79630)(86, 27.78539)(91, 28.72074)(96, 29.60526)};
   \addlegendentry{B = 1}

\addplot[color=yellow,mark=o] coordinates {
(1, 0.50000)(6, 2.99863)(11, 5.48822)(16, 7.95838)(21, 10.39756)(26, 12.79346)(31, 15.13346)(36, 17.40519)(41, 19.59693)(46, 21.69811)(51, 23.69968)(56, 25.59435)(61, 27.37676)(66, 29.04357)(71, 30.59338)(76, 32.02659)(81, 33.34524)(86, 34.55276)(91, 35.65367)(96, 36.65337)};
   \addlegendentry{B = 2}

\addplot[color=red,mark=triangle*] coordinates {
    (1, 0.50000)(6, 3.00000)(11, 5.49999)(16, 7.99988)(21, 10.49932)(26, 12.99732)(31, 15.49172)(36, 17.97839)(41, 20.45024)(46, 22.89618)(51, 25.30028)(56, 27.64139)(61, 29.89366)(66, 32.02817)(71, 34.01569)(76, 35.83005)(81, 37.45160)(86, 38.86960)(91, 40.08321)(96, 41.10089)};
   \addlegendentry{B = 4}

\addplot[color=cyan,mark=x,]    coordinates {
    (1, 0.50000)(6, 3.00000)(11, 5.50000)(16, 8.00000)(21, 10.50000)(26, 13.00000)(31, 15.50000)(36, 17.99997)(41, 20.49984)(46, 22.99928)(51, 25.49719)(56, 27.99033)(61, 30.47028)(66, 32.91774)(71, 35.29416)(76, 37.53424)(81, 39.54817)(86, 41.24314)(91, 42.55981)(96, 43.49918)};
   \addlegendentry{B = 8}

\addplot[color=black,mark=triangle,]coordinates {
    (1, 0.50000)(6, 3.00000)(11, 5.50000)(16, 8.00000)(21, 10.50000)(26, 13.00000)(31, 15.50000)(36, 18.00000)(41, 20.50000)(46, 23.00000)(51, 25.50000)(56, 27.99999)(61, 30.49992)(66, 32.99933)(71, 35.49503)(76, 37.96951)(81, 40.34821)(86, 42.41841)(91, 43.86681)(96, 44.60602)};
   \addlegendentry{B = 16}
\end{axis}
\end{tikzpicture}
\caption{Lower bounds on the competitive ratio of every randomized algorithm when $W = V = 10$ where the number of unknown packets arriving in a time slot varies, for different values of $B$ }
\label{fig:tight_lower_bnds}
\end{figure}

In the next section we use the insight obtained from the analysis in the current section to identify several algorithmic concepts appropriate for the problem of buffering with limited knowledge.

\section{Algorithmic Concepts}
\label{sec:algorithmic_conecpts}
In this section we describe the algorithmic concepts underlying our proposed algorithms for
dealing with scenarios of limited knowledge.

\paragraph*{Random selection}
\revchange{For obtaining a good competitive ratio we would like to avoid a scenario where \opt\ successfully transmits a bulk of ``good'' packets, which are originally unknown, while having the online algorithm discard all these packets. This translates to assuring each arriving $U$-packet has some minimal probability of being accepted and parsed.
}
\paragraph*{Speculatively Admit}
Competitive algorithms must ensure they retain throughput from both $K$-packets and $U$-packets.
Furthermore, once a $U$-packet is accepted, there is a high motivation to reveal its
characteristics as soon as possible, thus making educated decisions in the next cycles.

We therefore propose to {\em speculatively} over-prioritize unknown packets over known packets in certain cycles. \revchange{We refer to
the act of over-prioritizing an unknown packet $p$ in some cycle $t$ as {\em admitting} $p$. Respectively, we refer to such a cycle $t$ as an {\em admittance cycle}, and to such a packet $p$ as an {\em
admitted packet}}.

\paragraph*{Classify and randomly select}
Intuitively, as unknown packet characteristics are drawn from a wider range of values, the task of maximizing throughput becomes harder, especially when compared to the optimal throughput possible. To deal with this diversity, we apply a Classify and Randomly Select scheme~\cite{RCnS}.
\revchange{
This approach is based on the following notion: Assume we have an algorithm $\algc$ which is guaranteed to be $c$-competitive if traffic is sufficiently uniform, i.e., for cases where traffic characteristics are within some well-defined range of values. Given some arbitrary input sequence, which might be highly heterogeneous, we virtually partition the sequence of arriving packets into $N>1$ disjoint sub-sequences, which we refer to as {\em classes}, such that each class is sufficiently uniform, i.e., for any specific class $1 \leq i \leq N$ the characteristics of packets corresponding to class $i$ are within some well-defined range of values (as prescribed by $\algc$). The scheme then dictates selecting one of the classes {\em uniformly at random}, and applying $\algc$ to this class, while ignoring all packets corresponding to other classes. One then shows that the overall competitive ratio of this randomized approach is $O(N\cdot c)$-competitive for the overall input sequence.
}

\paragraph*{Alternate between fill $\&$ flush}
This paradigm is especially crucial in cases of limited information. The main motivation for this approach is that whenever a ``good'' buffer state is identified, the algorithm should focus all its efforts on monetizing the current state, maybe even at the cost of dropping packets indistinctly.
\revchange{In terms of buffer management and scheduling, this translates to defining some periods, in which the algorithm processes and transmits all the packets in its buffer, even at the cost of discarding all the arrivals. If these flush periods are short enough, the algorithm gains the high throughput from flushing its buffer, yet without compromising too much throughput due to having packets discarded during the flush.}

\section{Competitive Algorithms}
\label{sec:algorithms}
In this section we present a basic competitive online algorithm for the problem of buffering and scheduling with limited knowledge. We first provide a high-level description of our algorithm, and then turn to specify its details and
analyze its performance.

For simplicity of analysis and algorithm presentation, we assume that the set of possible values of $W$ and $V$ -- the work and profit per packet, respectively -- are known to the algorithm in advance.
In Sections~\ref{sec:improved_algorithms} and~\ref{sec:Practical_Implementation} we show how to remove this assumption without harming the performance of our algorithm, and present several improved variants of this algorithm.
We further note that neither of our proposed solutions require knowing the value of $M$ -- the maximum number of unknown packets arriving in a single cycle -- in advance.

\subsection{High-level Description of Proposed Algorithm}
\label{sec:algorithms_high_level_description}

Our algorithm is designed according to the algorithmic concepts presented in Section \ref{sec:algorithmic_conecpts} as follows.

\paragraph*{Randomly select and speculatively admit}
In every cycle $t$ during which a $U$-packet arrives, the algorithm picks
$t$ as an admittance cycle with some probability $r$
(to be determined in the sequel). In every cycle chosen
as an
admittance cycle, the algorithm picks exactly one of the $U$-packets arriving at $t$ to serve as the {\em admitted} packet.  This $U$-packet is chosen uniformly at random out of all $U$-packets arriving at $t$.
At the end of the arrival step, the algorithm schedules the admitted $U$-packet (if one exists) for processing,
hence {\em parsing} the packet. We note that if no such $U$-packet exists, or if $t$ is not an admittance cycle, then the algorithm may only accept known arriving packets, and would eventually schedule the top-priority packet residing in the Head-of-Line (HOL) for processing. The exact notion of priority will be detailed later.

\paragraph*{Classify and randomly select}
We implicitly partition the possible types of arriving packets into classes $C_1, C_2, \ldots C_m$; the criteria for partitioning and the exact value of $m$ will be specified later. Our algorithm picks a single {\em selected} class, uniformly at random from the $m$ classes.
Our goal is to provide {\em guarantees} on the performance of our proposed algorithm for packets belonging to the selected class, which is henceforth denoted $\Cs$.
Packets which belong to the selected class are referred to as {\em $\Cs$-packets}.
Following our previously introduced notation, known (unknown) packets that belong to the selected class, i.e.,
$\Cs$-packets for which their attributes are known (unknown), are denoted as {\em $\CsK$-packets} ({\em $\CsU$-packets}).

Focusing solely on packets belonging to $\Cs$ may seem like a questionable choice, especially if there are few packets arriving which belong to this class, or if the characteristics of packets belonging to this class are poor (e.g., they have low profit and require much work). However, this naive description is meant only to simplify the analysis.
In Section~\ref{sec:improved_algorithms} we show how to remedy
this naive approach in order to deal with these apparent shortcomings, while keeping the analytic guarantees intact.

\paragraph*{Alternate between fill $\&$ flush}
Our algorithm will be alternating between two states: the {\em fill} state, and the {\em flush} state.
We define an algorithm to be {\em \hfull} if its buffer is filled with known $\Cs$-packets.
Once becoming \hfull, our algorithm switches to the flush state, during which it discards all arriving packets and
continuously processes queued packets. Once the buffer empties, the algorithm returns to the fill phase.
Again, in Section~\ref{sec:improved_algorithms} we show how to improve upon this naive simplified approach.

\subsection{A General Classify and Randomly Select Mechanism}\label{sec:RC&S}
We now turn to explain the fundamentals of the classifying mechanism of our algorithm.

For each packet $p$ we assign a \emph {work-class} $C\W_i$, and denote the set of potential characteristic values within class $C\W_i$ by $X\W_i$.
 Let $\deltaW$ denote the maximal ratio between the work values of two packets, which belong to the same work-class. Similarly, for each packet $p$ we assign a \emph {profit-class} $C\P_i$, and denote the set of potential characteristic values within class $C\P_i$ by $X\P_i$. Let $\deltaP$ denote the maximal ratio between the profits of two packets, which belong to the same profit-class. Throughout our analysis, we will use $\deltaP$ and $\deltaW$ which are both constants.

Denote by $\mW$ and $\mP$ the number of work-classes and profit-classes, respectively.
We say a packet $p$ is of {\em combined-class} $C_{(i,j)}$  if it is of work-class $C\W_i$ and of profit-class $C\P_j$.
Note that in terms of work, the class to which a packet $p$ belongs is defined statically by the total work of
$p$, and does not depend upon its remaining processing cycles, which may change over time.

Upon initialization, the algorithm selects a class by picking $i^* \in \set{1,\ldots,\mW}$ and $j^* \in
\set{1,\ldots,\mP}$, each chosen uniformly at random. Then, the selected combined-class is $\Cs =
C_{(i^*,j^*)}$.

We will later define several ways to partition the packets into classes, each tailored and optimized for some specific scenarios of possible work and profit values.

\subsection{The \SAM\ Algorithm}
\label{sec:SAM}
We now describe the details of our algorithm, Speculatively Admit (\SAM), depicted in Algorithm~\ref{alg:SAM}.
The pseudo-code in Algorithm~\ref{alg:SAM} uses the procedures \emph{UpdatePhase()}, \emph{SortBuf()},
and \emph{\makeroom}, whose pseudo-code appears in Algorithms~\ref{alg:UpdatePhase},~\ref{alg:SortBuf} and~\ref{alg:MakeRoom}, respectively. The procedure \makeroom\ is destined to assure a free space for a high-priority arriving packet, even at the cost of pushing-out and dropping a lower-priority packet from the tail of the buffer, if the buffer is full.

\newcommand\AlgPhase[1]{%
\hspace*{\dimexpr-\algorithmicindent-2pt\relax}%
\vspace*{-.5\baselineskip}\Statex\hspace*{-\algorithmicindent}{\em #1}%
\vspace*{-.9\baselineskip}\Statex\hspace*{\dimexpr-\algorithmicindent-2pt\relax}%
}

\begin{algorithm}[t!]
\caption {UpdatePhase()}\label{alg:UpdatePhase}
\begin{algorithmic}[1]
\If {buffer is empty}
    \State {\em phase} = fill
\ElsIf {buffer is \hfull}
    \State {\em phase} = flush
\EndIf
\Comment {if buffer is neither empty  nor \hfull, {\em phase} is unchanged.}
\end{algorithmic}
\end{algorithm}

\begin{algorithm}[t!]
\caption {SortBuf()}\label{alg:SortBuf}
\begin{algorithmic}[1]
\State sort queued packets as follows: admitted packet first; $\CsK$-packets next; rest of the packets last; break ties by FIFO
\end{algorithmic}
\end{algorithm}

\begin{algorithm}[t!]
\caption {MakeRoom()}\label{alg:MakeRoom}
\begin{algorithmic}[1]
\If {the buffer is full}
	\State SortBuf()
	\State drop a packet from the tail
\EndIf
\end{algorithmic}
\end{algorithm}

\begin{algorithm}[t!]
\caption {\SAM: at every time slot $t$ after transmission} \label{alg:SAM}
\begin{algorithmic}[1]
\Statex
\AlgPhase{Arrival Step:}
\State {\em phase} = UpdatePhase()
\label{alg:SAM:line:UpdatePhaseAtArrivalBegin}
\State {\em admittance} = true w.p. $r$
\label{alg:SAM:line:DecideAdmittance}
\While {{\em phase} $==$ fill \AND exists arriving packet $p$} \label{alg:SAM:line:while_begin}
    \If {$p$ is a $\CsK$-packet}\label{alg:SAM:if_is_Gk}
	   \If {there are $B-1$ $\CsK$-packets in the buffer}\label{alg:SAM:if_B_minus_1}
	       \State drop admitted packet if exists        \label{alg:SAM:drop_ap}
       \EndIf    \label{alg:SAM:if_B_minus_1_endif}
	   \State \makeroom\label{alg:SAM:MakeRoom_by_Gk}
       \State accept $p$ \label{alg:SAM:accept_Gk}
    \ElsIf {$p$ is unknown AND \emph{admittance}}\label{alg:SAM:if_p_is_U}
        \If {$\Autp=1$}\label{alg:SAM:if_reservoir}
    	   \State \makeroom \label{alg:SAM:MakeRoom_by_ap}
            \State mark $p$ as admitted
            \State accept $p$ \label{alg:SAM:accept_ap}
        \Else
    	   \State w.p. $1/\Autp$, swap the admitted packet with $p$.
            \label{alg:SAM:line:reservoir}
        \EndIf\label{alg:SAM:if_reservoir_endif}
    \EndIf
    \If {buffer is not full} \label{alg:SAM:line:if_full}
        \State accept $p$ \label{alg:SAM:line:greedily_accept}
    \EndIf
	\State {\em phase} = UpdatePhase()
	\label{alg:SAM:line:UpdatePhase_in_arrival}
	\State SortBuf()
    \label{alg:SAM:line:sort_queue_in_arrival}
\EndWhile \label{alg:SAM:line:while_end}
\Statex
\AlgPhase{Processing Step:}
\State process HoL-packet
\label{alg:SAM:line:processHoL}
\State {\em phase} = UpdatePhase()
\label{alg:SAM:line:UpdatePhase_after_processing}
\State SortBuf()
\label{alg:SAM:line:sort_queue_in_processing}
\end{algorithmic}
\end{algorithm}

Once in the arrival step, algorithm \SAM\ updates its phase (line \ref{alg:SAM:line:UpdatePhaseAtArrivalBegin}).
In each cycle, the algorithm tosses a coin with some probability $r$, to be determined later, to decide whether this is an \emph{admittance cycle}, namely, a cycle in which the algorithm may admit an unknown packet (line~\ref{alg:SAM:line:DecideAdmittance}).
If the phase is flush, the algorithm skips the while loop (lines
\ref{alg:SAM:line:while_begin}-\ref{alg:SAM:line:while_end}), thus discarding all arriving packets.

If the phase is fill, which in particular implies that the buffer is not \hfull, the algorithm accepts every arriving $\CsK$-packet (lines~\ref{alg:SAM:if_is_Gk}-\ref{alg:SAM:accept_Gk}).
For assuring a free slot for the arriving $\CsK$-packet, the algorithm calls \makeroom\ (line~\ref{alg:SAM:MakeRoom_by_Gk}) before accepting the packet (line~\ref{alg:SAM:accept_Gk}).
The if-clause in lines~\ref{alg:SAM:if_B_minus_1}-\ref{alg:SAM:if_B_minus_1_endif} handles the special case where there are already $B-1$ $\CsK$-packets in the buffer; in this special case, after accepting the arriving $\CsK$-packet, the buffer will become \hfull, and therefore it should stop admitting packets.

If the phase is fill and this is an \emph{admittance} cycle (line~\ref{alg:SAM:if_p_is_U}), the algorithm admits a single $U$-packet arriving in this cycle, if such a packet exists.
In lines~\ref{alg:SAM:if_reservoir},\ref{alg:SAM:line:reservoir}, $\Autp$ denotes the number of $U$-packets which arrive in cycle $t$ by the arrival of packet $p$, including $p$ itself. Lines~\ref{alg:SAM:if_reservoir}-\ref{alg:SAM:if_reservoir_endif}
essentially perform a reservoir sampling~\cite{Reservoir}, which imply that the admitted $U$-packet is chosen uniformly at random out of all $U$-packets arriving in this cycle.

Finally, if the buffer is not full, the algorithm greedily accepts every arriving packet (lines \ref{alg:SAM:line:if_full}-\ref{alg:SAM:line:greedily_accept}).

While in the processing step, the algorithm simply processes the top-priority packet in the buffer (line~\ref{alg:SAM:line:processHoL}). Finally, the algorithm updates its phase and sorts the queued packets each time it either accepts or processes a packet (lines
\ref{alg:SAM:line:UpdatePhase_in_arrival}-\ref{alg:SAM:line:sort_queue_in_arrival} and
\ref{alg:SAM:line:UpdatePhase_after_processing}-\ref{alg:SAM:line:sort_queue_in_processing}). Note that the marking of a packet as an ``admitted packet'' is \emph{cycle-based}, namely, once an admitted packet is processed, it is not considered ``admitted'' anymore.
\revchange{To better understand \SAM, please refer to ~\ref{App:run_Example}, showing a running example of the algorithm.}

\subsection{Performance Analysis}
\label{sec:SAMWP}

We now turn to show an upper bound on the performance of our algorithm (for $W, V > 1$), captured by the following theorem (see~\ref{thm:SAMWP:proof} for the proof):

\begin{theorem}\label{thm:SAMWP}
\SAMWP\ is
$
O
\left(
\left[
\frac{M}{r} + \deltaW \cdot \deltaP
\right]
\cdot
\mW \cdot \mP
\right)$
-competitive.
\end{theorem}

\revchange {Theorem~\ref{thm:SAMWP} shows an inverse linear dependency of the competitive ratio on the
probability of choosing a cycle as an admittance cycle $r$. Thus, the best competitive ratio is attained for $r=1$, i.e., every cycle where $U$-packets arrive should be an admittance cycle.
In practical scenarios, however, one might want to be more conservative in choosing admittance cycles. E.g., one might choose $r<1$ so as to allow non-parsing cycles even when $U$-packets arrive, thus speeding up the processing of $\CsK$-packets. If one indeed chooses $r=1$, randomization should be maintained only for choosing the specific $U$-packet to be admitted, and the choice of the selected class.
We further explore the effect of the choice of parameter $r$ in Section~\ref{sec:simulations}.}

In the special cases of homogeneous work values (homogeneous profit values), we assign $\deltaW = \mW = 1$ ($\deltaP = \mP = 1$, resp.) in the upper bound implied by Theorem \ref{thm:SAMWP}, and obtain the following corollary:
\begin{corollary}\begin{inparaenum}[(a)]
\hfill \break
\item In the special case of homogeneous work values, \SAMWP\ is $O\left(\left(\frac{M}{r} + \deltaP \right) \cdot \mP\right)$-competitive.
\\\item In the special case of homogeneous profit values, \SAMWP\ is $O\left( \left(\frac{M}{r} + \deltaW \right) \cdot \mW\right)$-competitive.
\end{inparaenum}
\end{corollary}

Lastly, we note that when all packets are known upon arrival, i.e. $M=0$, \SAMWP\ is $(\deltaW \cdot \deltaP \cdot \mW \cdot \mP)$-competitive (see ~\ref{thm:SAMWP:proof}).

\subsection{Concrete Classification Mechanisms}\label{sec:RC&S_concrete}
We now show various classify and randomly select mechanisms, which are tailored and optimized for different scenarios, depending on the profit and work values.

\paragraph{A linear classification} When a characteristic consists of a small set of potential values, we let each class include a single value of this characteristic. As a result, the competitive ratio of the algorithm is linearly depended upon the number of distinct potential value of the respective characteristic. For instance, when the set of potential work values
is small, we let each potential work value define a class. As a result, the competitive ratio of \SAMWP, implied by Theorem~\ref{thm:SAMWP}, is linearly depended upon the number of distinct work values, captured by the parameter $\mW$. Note that in this case we have $X\W_i = \set{w_i}$, implying that $\deltaW$, the max-to-min ratio of values within $X\W_i$, is 1.

\paragraph{A logarithmic classification} When the set of potential values of a characteristic is large, letting each value define a unique class results in a poor competitive ratio. Therefore, in such cases we use a logarithmic-scaled class partitioning as follows. We say that a packet $p$ is of a certain class (either work- or profit-) $i$ if its corresponding value is in the
interval
\begin{equation} \label{Eq:log_classes}
  X_i =\begin{cases}
               [1,2] & i=1 \\
               [2^{i-1}+1, 2^i] & i > 1.
            \end{cases}
\end{equation}

In particular, using the above partition packets into classes, we obtain that $\deltaP = \deltaW = 2$, $\mP = \log_2V$ and $\mW = \log_2W$. Using Theorem~\ref{thm:SAMWP}, we obtain the following corollary:

\begin{corollary}
\SAMWP\ is $O\left(\frac{M}{r} \log_2 W \log_2 V\right)$-competitive.
\end{corollary}

We note that if we know the number of distinct values for each characteristics and the values of $W$ and $V$, we can choose the appropriate classification scheme and have $\mW$ to be the minimum between $\log_2 W$, and the number of distinct work values; and have $\mP$ to be the minimum between $\log_2 V$, and the number of distinct profit values. Moreover, in any of our classification schemes, $\deltaW, \deltaP \leq 2$.

\section{Improved Algorithms}
\label{sec:improved_algorithms}
Algorithm \SAM\ selects a single class uniformly at random so that the characteristics of
packets on which it focuses, namely, $\Cs$-packets, differ by at most a constant factor. This gives the sense of ``uniformity'' of traffic within the class being targeted, which in turn reduces the variability of characteristics of packets on which the algorithm focuses.
However, in practice there are various cases where the strict decisions made by \SAM\ can be relaxed without harming its competitive performance guarantees. In practice, such relaxations actually allow obtaining a throughput far superior to
that of \SAM.
In what follows we describe such modifications, which we incorporate into our improved algorithm, \SAO, and prove that all our performance guarantees for \SAM\ still hold for \SAO.

\paragraph*{Class closure}
Recall the partitioning of packets into classes, described in Section~\ref{sec:RC&S}, namely,
 $\set{C_{(i,j)} | i=1,\ldots,\mW, j=1,\ldots,\mP}$. We let the {\em $(i,j)$-closure class} be defined as
 $C^*_{(i,j)}=\bigcup_{i'\leq i, j'\geq j} C_{(i',j')}$.

This definition means that the work of any packet in $C^*_{(i,j)}$
is within a ratio of at most $\deltaW$ of the work of any packet in $C_{(i,j)}$, and similarly for the profit of any packet in $C^*_{(i,j)}$.
Formally, for any packets $p \in C_{(i,j)}$ and  $p^* \in C^*_{(i,j)}$, $w(p^*) \leq \deltaW \cdot w(p)$ and $ v(p^*) \geq \frac{v(p)}{\deltaP}$.

We let \SAOWP\ denote the algorithm where the selected class $\Cs$ is chosen to be $C^*_{(i,j)}$, for some values of
$i,j$ chosen uniformly at random from the appropriate sets.
A simple substitution argument shows that thus picking $C^*_{(i,j)}$ by \SAOWP, instead of selecting $C_{(i,j)}$ as done
in \SAMWP, leaves the analysis detailed in Section~\ref{sec:SAMWP} intact.

\paragraph*{Fill during flush (pipelining)}
Algorithm \SAM\ was defined such that no arriving packets are ever accepted during the flush phase. This enables the partitioning of time into disjoint intervals (determined by \SAM's buffer being empty et the end of such an interval), and applying the comparison of performance of \opt, on the one hand, and \SAM, on the other hand, independently for each interval. In practice, however, allowing the acceptance of packets during a flush phase
cannot harm the analysis, nor the actual performance, if this is done prudently:
packets which arrive during the flush phase are accepted according to the same priority suggested by the algorithm's behavior in the fill phase.
Furthermore, the algorithm stores in the buffer packets which arrive during the flush phase, but never schedules them for processing before it successfully transmits all $B$ packets that were stored in the buffer when it turned \hfull.

\paragraph*{Improved scheduling}
\SAM\ sorts the queued packets in $\CsK$-first order.
For simplicity of presentation, we assumed in Section~\ref{sec:algorithms} that within the set of $\CsK$-packets, as well as within the set of non-$\CsK$-packets, packets are  internally ordered by FIFO.
However, one may consider other approaches as well to performing such scheduling for each of these sets (while maintaining $\CsK$-first order between the sets). We consider specifically the following methods:
\begin{inparaenum}[(i)]
\item FIFO,
\item $W$-then-$V$, which orders packets by a non-decreasing order of remaining work, and breaks ties by non-increasing order of profit, and
\item non-increasing order of packet {\em effectiveness}, where the effectiveness of a packet is defined as its
    profit-to-work ratio.
\end{inparaenum}

We emphasize that the packet scheduled for processing during an admittance cycle remains a $U$-packet, which is selected uniformly at random from the arriving $U$-packets at this cycle.
All the \emph{non}-admitted $U$-packets, however, are located at the tail of the queue, thus representing the fact that their priority is lower than that of every known packet.
By applying different scheduling regimes, we obtain different flavors of \SAO.

The following Theorem shows that the performance of all flavors of \SAO\ is at least as good as the performance of \SAM.
\begin{theorem}\label{thm:SAOWP}
\SAOWP\ is
$
O
\left(
\left[
\frac{M}{r} + \deltaW \cdot \deltaP
\right]
\cdot
\mW \cdot \mP
\right)$
-competitive.
\end{theorem}

For the proof, see~\ref{thm:SAOWP:proof}. We study the performance of the various flavors of \SAO\ in Section \ref{sec:simulations}.

\section{Practical Implementation}\label{sec:Practical_Implementation}
While presenting our basic algorithm in Section \ref{sec:algorithms}, we assumed for simplicity that the values of $W$
and $V$ -- the maximal work and profit per packet, respectively -- are known to the algorithm in advance.
We now show how to relax these assumptions without harming the performance of our algorithms.

We refer to an algorithm implementation that does
not know these values in advance as a {\em values-oblivious} algorithm, and to an algorithm implementation that knows the values of $W$ and $V$ in advance as a {\em values-aware} algorithm. We will show that a values-oblivious algorithm can obtain a performance which is no worse than that of a values-aware algorithm, even if the values-aware algorithm knows not only $W$ and $V$, but also the concrete classes in which packets will arrive.

Our implementation of a values-oblivious algorithm is based on an application of reservoir sampling \cite{Reservoir} on
classes revealed during packet arrivals, as we will detail shortly.
A new class is revealed either due to the arrival of a $K$-packet $p$, or due to a $U$-packet $q$ being parsed,
corresponding to a class previously unknown to the algorithm. We call such an event an {\em uncovering of a new
class}.

The values-oblivious algorithm implementation performs the following alongside all decisions made by the values-aware algorithm:
Before the arrival sequence begins we initiate a counter $N$ of known classes to be $N=0$.
Upon the uncovering of a new class at $t$ the algorithm increments $N$ by one (to reflect the updated number of
known classes), and replaces the previously selected class with the new class with probability $1/N$.

As the above procedure essentially performs a reservoir sampling on the collection of classes known to the algorithm,
it essentially implements the selection of a class uniformly at random among all {\em a posteriori} known
classes~\cite{Reservoir}.

It therefore follows that the distribution of the packets corresponding to the eventual selected class (after the sequence ends) handled by the values-oblivious algorithm is identical to the distribution of the packets handled by the values-aware algorithm. Therefore the expected performance of the values-oblivious algorithm is lower bounded by the expected performance of the values-aware algorithm.
We note that the implementation of the values-oblivious algorithm can be applied to any of the variants described in our previous sections.

\section{Simulation Study}
\label{sec:simulations}
In this section we present the results of our simulation study intended to validate our theoretical results, and
provide further insight into our algorithmic design.
\revchange{Our choice of distributions for the parameters of the traffic characteristic enables us to evaluate our algorithms performance in a wide range of settings. These choices, as we show in the sequel, are also motivated by the properties of real-world traffic.}

\subsection{Simulation Settings}\label{sec:simulation_set}
We simulate a single queue in a gateway router which handles a bursty arrival sequence of packets with high work
requirements (corresponding, e.g., to IPSec packets, requiring AES encryption/decryption) as well as packets with
low work requirements (such as simple IP packets requiring merely IPv4-trie processing).
Arriving packets also have arbitrary profits, modeling various QoS levels.

Our traffic is generated by a Markov modulated Poisson process (MMPP)
with two states, LOW and HIGH, such that the burst during the HIGH state
generates an average of 10 packets per cycle, while the LOW state generates an average of only $0.5$ packet per
cycle. The average duration of LOW-state periods is a factor $W$ longer than the average duration of HIGH-state
periods. This is targeted at allowing some traffic arriving during the HIGH-state to be drained during the LOW-state.

In our simulations, we do not deterministically bound the maximum number, $M$, of $U$-packets arriving in a cycle, but rather control the expected intensity of $U$-packets by letting each arriving packet be a $U$-packet with some probability $\alpha \in [0,1]$. We thus obtain that the expected number of $U$-packets per cycle during the HIGH state is $10 \alpha$.

In real-life scenarios, the maximum work, $W$, required by a packet, is highly implementation-depended. It depends on the specific hardware, processing elements, and software modules.
However, several works which investigated the required work on typical tasks~\cite{ramaswamy09, salehi09, salehi12} indicate that $W$ is two orders of magnitude larger than the work required for doing an IPv4-trie search or classification of a packet. We refer to IPv4-trie search or classification of a packet as the baseline unit of work, captured by our notion of ``parsing''. We therefore set the maximum work required by a packet to $W=256$ throughout this section.
As the potential set of characteristics is large, we use a logarithmic classification scheme (recall Section~\ref{sec:RC&S_concrete}).

Determining the maximum profit, $V$, associated with a packet, is a challenging task. This value depends both on implementation details, as well as on proprietary commercial and business considerations. In order to have a diverse set of values, which model distinct QoS requirements, we set the maximum profit associated with a packet to $V=16$ throughout this section.

The values $W=256$ and $V = 16$ imply a total of $8 \cdot 4 = 32$ potential classes for the algorithm to select from, respectively.
\revchange{The value of each characteristic for each packet is drawn from an approximation of a Pareto-distribution as follows. First, we randomly generate numbers, following a Pareto-distribution. Next, numbers are rounded, to get integer values. Finally, for disallowing values above the maximum (256 for work values and 16 for profit values), all the cases where the randomly generated values were above the maximum were truncated, namely, treated as if the generated value was exactly the maximal value.
The averages and standard deviations of the values obtained after this generation process are 17.97 and 22.22 for packet work, and 3.66 and 3.20 for packet profit.
The schematic probability distribution function of the characteristics values is depicted in Fig.~\ref{fig:PDF_of_values}.
Note the spike at its maximum, due to the truncation described above.}
Unless stated otherwise, we assume that $B=10$, $r = 1$ and each arriving packet is a $U$-packet with probability $\alpha = 0.3$. We thus obtain that the expected number of $U$-packets arriving during the HIGH state is $0.3 \cdot 10 = 3$ per cycle.

\setlength{\belowcaptionskip}{-7pt}

\setlength{\belowcaptionskip}{-7pt}

\begin{figure}[t]
\centering
\begin{tikzpicture}
	\begin{axis}[
		ybar,
		bar width=2ex,		
		ticks=none,
		xlabel=Value,
		xlabel style={at={(0.5,2ex)}},
		ylabel style={at={(4ex,0.5)}},
        legend style = {font=\scriptsize},
        label style={font=\scriptsize},
        tick label style={font=\scriptsize},
		xticklabels=empty,
		ylabel=PDF,
		yticklabels=empty]
	\addplot[fill=gray] plot coordinates {
		(1,0.2136)
		(2,0.2779)
		(3,0.1659)
		(4,0.1053)
		(5,0.0635)
		(6,0.0440)
		(7,0.0303)
		(8,0.0219)
		(9,0.0173)
		(10,0.0113)
		(11,0.0080)
		(12,0.0063)
		(13,0.0054)
		(14,0.0047)
		(15,0.0030)
		(16,0.0214)
	};
	\end{axis}
\end{tikzpicture}	
\caption{Probability distribution function of the characteristics values}
\label{fig:PDF_of_values}
\end{figure}
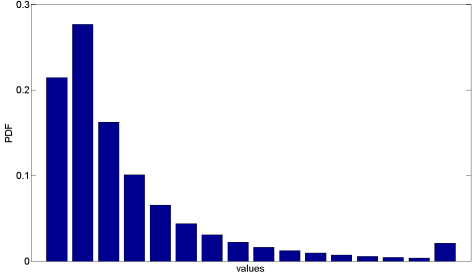

As a benchmark which serves as an upper bound on the optimal performance possible, we consider a relaxation of the offline problem as a knapsack problem. Arriving packets are viewed as items, each with its size (corresponding to the packet's work) and value (corresponding to the packet's profit). The allocated knapsack size equals the number of time slots during which packets arrive. The goal is to choose a highest-value subset of items which fits within the given knapsack size.
This is indeed a relaxation of the problem of maximizing throughput during the arrival sequence in the offline setting, since the knapsack problem is not restricted by any finite buffer size during the arrival sequence, nor by the arrival time of packets (e.g., it may ``pack'' packets even before they arrive).

We employ the classic 2-approximation greedy algorithm for solving the knapsack problem~\cite{williamson11design}, and
use its performance as an approximate upper bound on the performance of \opt.
For considering the additional profit which \opt\ may gain from packets which reside in its buffer at the end of the arrival sequence, we simply allow the offline approximation an additional throughput of $BV$ for free, which is an upper bound on the benefit it may achieve after the arrival sequence ends.

We compare the performance of studied algorithms by evaluating their {\em performance ratio}, which is the ratio between the algorithm's performance and that of our approximate upper bound on the performance of \opt.

We compare the performance of the following algorithms:
\begin{enumerate}
\item {\em FIFO}: A simple greedy non-preemptive FIFO discipline that simply accepts packets and processes each
    packet until completion, regardless of its required work or value.
\item {\em \SAM}: Algorithm \SAM, described in Section \ref{sec:algorithms}.
\item {\em \SAO\ FIFO}: Algorithm \SAO\ where priority ties are broken by FIFO order.
\item {\em \SAO\ $W$-Then-$V$}: Algorithm \SAO\ where priority ties are broken in non-decreasing order of remaining
    work, and further ties are broken in non-increasing order of profit. This variant is denoted by \SAO $W-V$ in Figures~\ref{fig:g_W}-\ref{fig:r}.
\item {\em \SAO\ EFFECT}: Algorithm \SAO\ where priority ties are broken in non-increasing order of their
    profit-to-work ratio.
\end{enumerate}

We recall that all the flavors of \SAO\ listed above maintain a $\CsK$-first order, and differ only in the internal
ordering \emph{within} each set (namely, within the set of $\CsK$-packets, as well as within the set of {\em
non}-$\CsK$-packets).

All flavors of \SAO\ described above employ the class-closure and the fill-during-flush modifications defined in
Section \ref{sec:improved_algorithms}.
For each choice of parameters we show the average of running 100 independently-generated traces of 10K packets each. In all our simulations the standard deviation was below 0.035.

\subsection{Simulation Results}
Figures~\ref{fig:g_W}-\ref{fig:r} show the results of our simulation study.
First we note that \SAM\ exhibits a very low performance ratio, similar to that of a simple FIFO (which disregards packets parameters altogether). This is due to the fact that \SAM\ focuses only on a specific class, which consists of a relatively small part of the input, and it thus
spends processing cycles on packets that would not be eventually transmitted.

For the variants of \SAO\ we consider, in all simulations the best scheduling policy is by non-increasing
effectiveness, followed by employing the $W$-then-$V$ approach. FIFO scheduling, in spite of it being simple and
attractive, comes in last in all scenarios.
This behavior is explained by the fact that both former scheduling policies in \SAO\ clear the buffer more
effectively once it is \hfull. The latter FIFO scheduling approach clears the buffer in an oblivious manner, and
therefore doesn't free up space for new arrivals fast enough.
We now turn to discuss each of the scenarios considered in our study.

\subsubsection{The Effect of Selected Class}
Our first set of results sheds light on the effect of the class selected by an algorithm on its performance.
Fig.~\ref{fig:g_W} shows the results where the selected profit-class $j^*$ is 1, which makes \SAO\ allow all profits, and the choice of work-class $i^*$ varies.
The most interesting phenomena is exhibited by \SAO\ FIFO. Its performance is very poor if the work-class may
contain packets requiring very little work. This is due to the fact that only a small fraction of the traffic
requires this little work, and the algorithm scarcely arrives at being \hfull. As a consequence, the algorithm
handles many low-priority packets, which are handled in FIFO order, giving rise to far-from-optimal decisions. The
algorithm steadily improves up to some point, and then its performance deteriorates fast as it assigns high-priority
to packets with increasingly higher processing requirements. In this case the algorithm becomes \hfull\ too
frequently, and allocates many processing cycles to low-effectiveness packets.
The maximum performance is achieved for $i^*=3$,
which implies that the algorithm flushes whenever its buffer is filled up with packets whose work is at most
$2^{i^*}=8$. This value suffices to allow the algorithm to prioritize a rather large portion of the arrivals
(recalling the Pareto distribution governing packet work-values), while ensuring the processing toll of
high-priority packet is not too large.
This strikes a (somewhat static) balance between the amount of work required by a packet, and its expected potential
profit.
The other variants of \SAO\ exhibit a gradually decreasing performance, due to
their higher readiness to compromise over the required work of packets they deem as high-priority traffic.
\SAM\ shows a similar performance deterioration, for a similar reason, when the selected work-class $i^*$ is increased from 1 up to 6. However, when increasing $i^*$ above 6, \SAM's performance increases again. This improvement is explained by the fact that, due to the Pareto-distribution of the work values, the number of packets which belong to each work-class rapidly diminishes when switching to work-class indices closest to the maximum of 8;
recall that \SAM\ over-prioritizes only packets which belong to a single randomly selected class, i.e., \SAM\ does not employ the class closure optimization (described in Section~\ref{sec:improved_algorithms}).
In such a case, \SAM\ is coerced into processing also packets which do not belong to the selected class -- namely, packets with {\em lower} work -- which somewhat compensates for the poor choice of the work-class.
We verified this explanation by additional simulations (not shown here), in which the work-class of packets was chosen from the uniform distribution. In such a case, where there is an abundance of packets from every possible work-class, the performance of \SAM\ consistently degrades with the increase of $i^*$, which implies a poorer choice of work-class.

Similar phenomena are exhibited in Fig.~\ref{fig:g_V}, where we consider the effect of the profit-class $j^*$
selected by an algorithm on its performance. In this set of simulations all work-values were allowed (i.e., the
selected work-class is 8).
In this scenario the performance
of all algorithms improves as the selected profit-class index increases, and the
algorithms are able to better restrict their focus on high profit packets as the packets receiving high-priority.
We note the fact that \SAO\ FIFO and regular FIFO have a matching performance in the case the selected profit-class
is 1, since in this case \SAO\ FIFO is identical to plain FIFO (since it simply indiscriminately accepts and processes all incoming packets in FIFO order).

In subsequent results described hereafter, we fix both the work-class and the profit-class to be 3, which represents
a mid-range class for both the profit and the work.

\setlength{\belowcaptionskip}{-7pt}

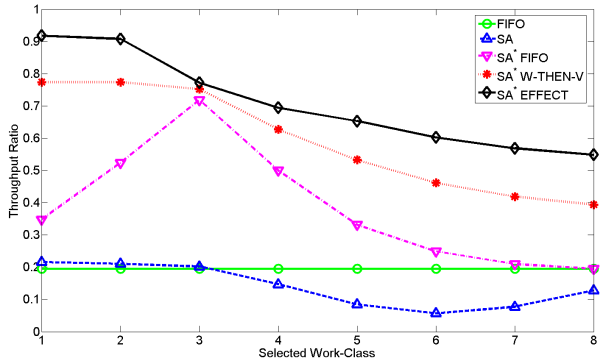
\begin{figure}[t]
\centering
\begin{tikzpicture}
	\begin{axis}[
		legend style={at={(0.5,1.18)},anchor=north,legend columns=-1,font=\scriptsize},
	    label style={font=\scriptsize},
	    tick label style={font=\scriptsize},
		ylabel=Throughput ratio,
		xlabel=Selected work-class,
		ytick={0,0.1,0.2,0.3,0.4,0.5,0.6,0.7,0.8,0.9,1},
		ymin=0,
		ymax=1,
	    ymajorgrids=true,
	    grid style=dashed,
		xtick=data]
	\addplot[color=blue,mark=x] coordinates {
		(1,0.1930)(2,0.1930)(3,0.1930)(4,0.1930)(5,0.1930)(6,0.1930)(7,0.1930)(8,0.1930)
	};
	\addlegendentry{FIFO}
	\addplot[color=red,mark=triangle*] coordinates {
		(1,0.2137)(2,0.2103)(3,0.1993)(4,0.1455)(5,0.0832)(6,0.0559)(7,0.0729)(8,0.1262)
	};
	\addlegendentry{SA}
	\addplot[color=yellow,mark=o] coordinates {
		(1,0.3456)(2,0.5201)(3,0.7123)(4,0.4954)(5,0.3280)(6,0.2452)(7,0.2096)(8,0.1926)
	};
	\addlegendentry{SA$^*$ FIFO}
	\addplot[color=cyan,mark=*] coordinates {
		(1,0.7671)(2,0.7671)(3,0.7453)(4,0.6208)(5,0.5262)(6,0.4571)(7,0.4113)(8,0.3928)
	};
	\addlegendentry{SA$^*$ $W$-$V$}
	\addplot[color=black,mark=square*] coordinates {
		(1,0.9086)(2,0.8988)(3,0.7655)(4,0.6867)(5,0.6442)(6,0.5940)(7,0.5635)(8,0.5425)
	};
	\addlegendentry{SA$^*$ EFFECT}
	\end{axis}
\end{tikzpicture}	
\caption{Effect of chosen work-class $i^*$}
\label{fig:g_W}
\end{figure}

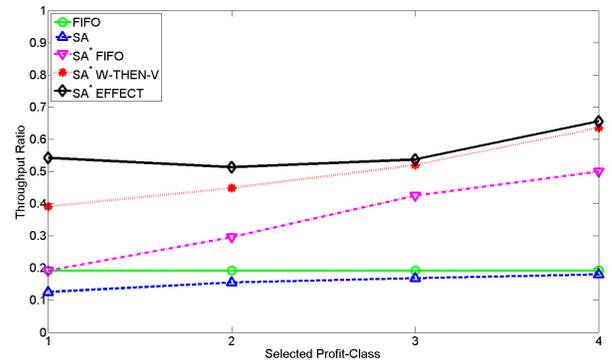
\begin{figure}[t]
\centering
\begin{tikzpicture}
	\begin{axis}[
		legend style={at={(0.5,1.18)},anchor=north,legend columns=-1,font=\scriptsize},
	    label style={font=\scriptsize},
	    tick label style={font=\scriptsize},
		ylabel=Throughput ratio,
		xlabel=Selected profit-class,
		ytick={0,0.1,0.2,0.3,0.4,0.5,0.6,0.7,0.8,0.9,1},
		ymin=0,
		ymax=1,
	    ymajorgrids=true,
	    grid style=dashed,
		xtick=data]
	\addplot[color=blue,mark=x] coordinates {
	    (1,0.1919)(2,0.1919)(3,0.1919)(4,0.1919)
	};
	\addlegendentry{FIFO}
	\addplot[color=red,mark=triangle*] coordinates {
		(1,0.1255)(2,0.1554)(3,0.1686)(4,0.1803)
	};
	\addlegendentry{SA}
	\addplot[color=yellow,mark=o] coordinates {
		(1,0.1922)(2,0.2960)(3,0.4248)(4,0.4997)
	};
	\addlegendentry{SA$^*$ FIFO}
	\addplot[color=cyan,mark=*] coordinates {
		(1,0.3910)(2,0.4486)(3,0.5201)(4,0.6358)
	};
	\addlegendentry{SA$^*$ $W$-$V$}
	\addplot[color=black,mark=square*] coordinates {
		(1,0.5427)(2,0.5137)(3,0.5372)(4,0.6557)
	};
	\addlegendentry{SA$^*$ EFFECT}
	\end{axis}
\end{tikzpicture}
\caption{Effect of chosen profit-class $j^*$}
\label{fig:g_V}
\end{figure}

\subsubsection{The Effect of Missing Information}
Fig.~\ref{fig:M} illustrates the performance ratio of our algorithms as a function of the expected number of $U$-packets arriving during the HIGH state, where we vary the value of $\alpha$ from 0 to 1. This provides further insight as to the performance of each algorithm as a function of the intensity of unknown packets. We recall that for our choice of parameters, the values of $\alpha$ translate to having the expected number of unknown packets per cycle during the HIGH state vary from 0 to 10.
As one could expect, the performance ratio of \SAM\ and of all versions of \SAO\ degrades as the amount of uncertainty increases.

Finally, we study the intensity of exploring unknown packets, as depicted by the choice of parameter $r$ which
determines whether a cycle is an admittance cycle or not.
The results depicted in Fig.~\ref{fig:r} consider the case of high uncertainty, where $\alpha=1$, that is, all arriving packets are unknown.

Observe first the special case where $r=0$, which represents an extreme case, in which, although all arriving packets are unknown, our algorithms do not explore any new packets, and actually degenerate to a simple FIFO, and therefore exhibit identical performance.
Increasing the admittance probability $r$, however, yields a steady increase in performance, albeit with diminishing returns.
Similar results were obtained also when some of the packets are known, but with smaller marginal benefits. These results coincide with our analytic results, which further validate our algorithmic approach.

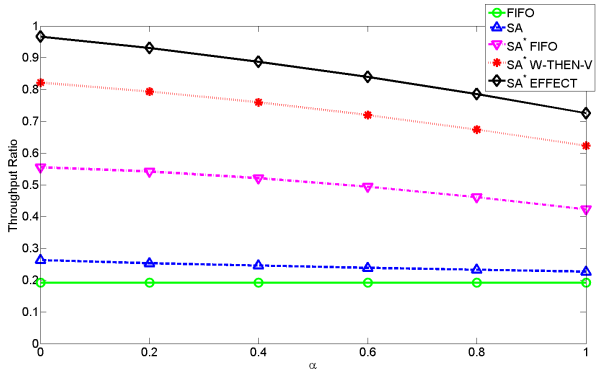
\begin{figure}[t]
\centering
\begin{tikzpicture}
	\begin{axis}[
		legend style={at={(0.5,1.18)},anchor=north,legend columns=-1,font=\scriptsize},
	    label style={font=\scriptsize},
	    tick label style={font=\scriptsize},
		ylabel=Throughput ratio,
		xlabel=$\alpha$,
		ytick={0,0.1,0.2,0.3,0.4,0.5,0.6,0.7,0.8,0.9,1},
		ymin=0,
		ymax=1,
	    ymajorgrids=true,
	    grid style=dashed,
		xtick=data]
	\addplot[color=blue,mark=x] coordinates {
	    (0,0.1924)(0.2,0.1924)(0.4,0.1924)(0.6,0.1924)(0.8,0.1924)(1,0.1924)
	};
	\addlegendentry{FIFO}
	\addplot[color=red,mark=triangle*] coordinates {
		(0,0.2627)(0.2,0.2527)(0.4,0.2455)(0.6,0.2384)(0.8,0.2326)(1,0.2268)
	};
	\addlegendentry{SA}
	\addplot[color=yellow,mark=o] coordinates {
		(0,0.5550)(0.2,0.5417)(0.4,0.5211)(0.6,0.4937)(0.8,0.4607)(1,0.4228)
	};
	\addlegendentry{SA$^*$ FIFO}
	\addplot[color=cyan,mark=*] coordinates {
		(0,0.8213)(0.2,0.7932)(0.4,0.7594)(0.6,0.7194)(0.8,0.6739)(1,0.6231)
	};
	\addlegendentry{SA$^*$ $W$-$V$}
	\addplot[color=black,mark=square*] coordinates {
		(0,0.9663)(0.2,0.9302)(0.4,0.8878)(0.6,0.8393)(0.8,0.7853)(1,0.7251)
	};
	\addlegendentry{SA$^*$ EFFECT}
	\end{axis}
\end{tikzpicture}	
\caption{Effect of expected number of $U$-packets during the HIGH state}
\label{fig:M}
\end{figure}

\begin{figure}[t]
\centering
\begin{tikzpicture}
	\begin{axis}[
		legend style={at={(0.5,1.18)},anchor=north,legend columns=-1,font=\scriptsize},
	    label style={font=\scriptsize},
	    tick label style={font=\scriptsize},
		ylabel=Throughput ratio,
		xlabel=$r$,
		ytick={0,0.1,0.2,0.3,0.4,0.5,0.6,0.7,0.8,0.9,1},
		ymin=0,
		ymax=1,
	    ymajorgrids=true,
	    grid style=dashed,
		xtick=data]
	\addplot[color=blue,mark=x] coordinates {
	    (0,0.1933)(0.2,0.1935)(0.4,0.1944)(0.6,0.1920)(0.8,0.1930)(1,0.1931)
	};
	\addlegendentry{FIFO}
	\addplot[color=red,mark=triangle*] coordinates {
		(0,0.1933)(0.2,0.1974)(0.4,0.2078)(0.6,0.2144)(0.8,0.2202)(1,0.2289)
	};
	\addlegendentry{SA}
	\addplot[color=yellow,mark=o] coordinates {
		(0,0.1933)(0.2,0.2251)(0.4,0.2738)(0.6,0.3242)(0.8,0.3746)(1,0.4260)
	};
	\addlegendentry{SA$^*$ FIFO}
	\addplot[color=cyan,mark=*] coordinates {
    	(0,0.1933)(0.2,0.2582)(0.4,0.4106)(0.6,0.5089)(0.8,0.5781)(1,0.6268)
	};
	\addlegendentry{SA$^*$ $W$-$V$}
	\addplot[color=black,mark=square*] coordinates {
		(0,0.1933)(0.2,0.2642)(0.4,0.4467)(0.6,0.5727)(0.8,0.6653)(1,0.7286)
	};
	\addlegendentry{SA$^*$ EFFECT}
	\end{axis}
\end{tikzpicture}	
\caption{Effect of admittance probability of $U$-packets $r$}
\label{fig:r}
\end{figure}
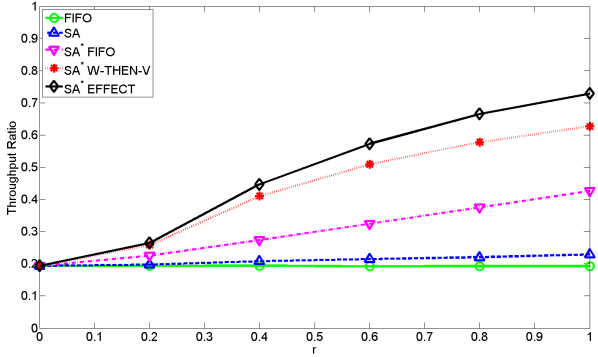

\section{Conclusions and Future Work}
\label{sec:conclusions}

We consider the problem of managing buffers where traffic has unknown characteristics, namely required processing
and profits.
\revchange{We show lower bounds on the competitive ratio of any online algorithm for the problem.}
We define several algorithmic concepts targeted at such settings, and develop several algorithms that
follow our suggested prescription.
\revchange{Our theoretical analysis shows that the competitive ratio of our algorithms is not far from the best competitive ratio any online algorithm can achieve.}
We validate the performance of our algorithms via simulation which further serves to elucidate our design criteria.
Our work can be viewed as a first step in developing fine-grained algorithms handling scenarios of limited knowledge
in networking environments for highly heterogeneous traffic.

Our work gives rise to a multitude of open questions, including:
\begin{inparaenum}[(i)]
\item closing the gap between our lower and upper bound for the problem,
\item applying our proposed approaches to other limited knowledge networking environments, and
\item devising additional algorithmic paradigms for handling limited knowledge in heterogeneous settings.
\end{inparaenum}


\section*{References}
\bibliography{TR_v2}

\appendix
\section{Preliminaries}
We now define some of the notation that will be used throughout the appendix.

For every cycle $t$ and packet type $\alpha$, we denote by $A^{\alpha}(t)$ the number of $\alpha$-packets that arrive in cycle $t$. For instance, $\Ak(t)$ ($\Au(t)$) denotes the number of $K$-packets ($U$-packets) which arrive in cycle $t$. This notation can be combined with the work and profit values of packets. For instance, $\Au_{(w,v)}(t)$ denotes the number of $U$-packets with work $w$ and profit $v$, which arrive in cycle $t$.

Our proofs involve a careful analysis of the expected profit of our algorithms from packets which arrive when it is either in the fill or the flush phase. Therefore, we now turn to define the exact notion of cycles belonging to either phase.
We say that an algorithm is in the flush phase in a specific cycle $t$ if it is in the flush state at the end of the arrival step of cycle $t$. If it's not in the flush phase in cycle $t$, then we say it is in the fill phase in cycle $t$. Denote by $P\mFILL$ and $P\mFLUSH$ the sets of cycles in which our algorithm is in the fill and flush phases, respectively.

For every packet type $\alpha$, we denote by $S_{\alpha}(t)$ the expected profit of the algorithm from
$\alpha$-packets which {\em arrive} in cycle $t$, and by
$S_{\alpha}=\sum_{t}S_{\alpha}(t)$ the overall
expected profit of \alg\ from $\alpha$-packets. We denote by $O_{\alpha}$ the expected profit of some optimal solution, \opt, from $\alpha$-packets. Again, these notations can be combined with previous notations. For instance, $O_{\CsU}(t)$ denotes the overall expected profit of \opt\ from $\CsU$-packets. Furthermore, $O\mFILL_{\CsU}$ denotes the expected profit of \opt\ from $\CsU$-packets which arrive during $P\mFILL$.

\section{Proof Of Theorem \ref{thm:SAMWP}}\label{thm:SAMWP:proof}
Our proof will follow from a series of propositions.
Initially, we aim to prove that \SAM\ successfully transmits every $\CsK$-packet which arrives during the fill phase, by showing that it never drops such a packet once it is accepted to the buffer.

\begin{proposition}
\label{prop:transmits_all_G_K}
\SAM\ successfully transmits every $\CsK$-packet which arrives during the fill phase.
\end{proposition}

\begin{proof}

We first note, that any $\CsK$-packet arriving during the fill phase (depicted by the while loop in lines~\ref{alg:SAM:line:while_begin}-\ref{alg:SAM:line:while_end}) is accepted (line~\ref{alg:SAM:accept_Gk}).

Next, we show that \SAM\ never drops a $\CsK$-packet which resides in its buffer.
We consider all cases where \SAM\ drops a packet from its buffer, and prove that it cannot be a $\CsK$-packet.

In line~\ref{alg:SAM:drop_ap}, \SAM\ drops an admitted packet, namely, a picked $U$-packet, and not a $\CsK$-packet.

In line~\ref{alg:SAM:MakeRoom_by_Gk}, \SAM\ performs the \makeroom\ procedure, which may result in dropping the last packet in the buffer. However, as this line dwells within the while loop of lines~\ref{alg:SAM:line:while_begin}-\ref{alg:SAM:line:while_end}, we know that the phase is fill, and therefore there are at most $B-1$ $\CsK$-packets in the buffer. Furthermore, if there are exactly $B-1$ $\CsK$-packets in the buffer, the if-clause in lines~\ref{alg:SAM:if_B_minus_1}-\ref{alg:SAM:if_B_minus_1_endif} assures that there is no admitted packet in the buffer. Hence, if the buffer is full, it contains at least one low-priority packet -- namely, a packet which is not admitted and not a $\CsK$-packet. After sorting the buffer, this low-priority, non-$\CsK$ packet, will be located in the tail of the queue and dropped.

\SAM\ may perform the \makeroom\ procedure also in line~\ref{alg:SAM:MakeRoom_by_ap}, if $\Autp = 1$. In this case, the arriving packet $p$ is the first $U$-packet arriving in this cycle -- and it is not admitted yet. As a result, there is no admitted packet in the buffer. Furthermore, as this line is executed during the fill phase (the while loop of lines~\ref{alg:SAM:line:while_begin}-\ref{alg:SAM:line:while_end}), there are at most $B-1$ $\CsK$-packets in the buffer. Hence, if the buffer is full, it contains at least one low-priority, non-$\CsK$-packet, which is the packet dropped.
\end{proof}

The following lemma shows that the overall number of $\Cs$-packets transmitted by \SAMWP\ is at least a significant fraction of the number of $\Cs$-packets accepted by an optimal policy during a fill phase.

\begin{lemma}\label{Cs_K_WP}
$S_{\Cs} \geq \frac{r}{M} O\mFILL_{\Cs}$.
\end{lemma}

\begin{proof}
Let $t$ denote a cycle in the fill phase, in which $U$-packets arrive. Then, with probability $r$ \SAMWP\ admits one $U$-packet, denoted $p$. As the algorithm implements reservoir sampling~\cite{Reservoir}, $p$ is picked uniformly at random out of at most $M$ unknown arrivals, and therefore the probability that $p \in \CsU$ is at least $A_{\CsU}(t)/M$.
As $p$ is parsed in the cycle of arrival, in the subsequent cycle it is known.
By Proposition~\ref{prop:transmits_all_G_K}, if $p$ is a $\CsK$-packet, then \SAMWP\ will eventually transmit $p$.
Recalling that $X\W_{i^*}$ and $X\P_{j^*}$ denote the ranges of the work and profit values within the selected work and profit class $C_{(i^*,j^*)}$ (see Sec.~\ref{sec:RC&S_concrete}), we conclude that
\begin{equation} \label{Eq:S_G_U_t}
S_{\CsU} (t) \geq \frac{r}{M} \sum_{ w \in X\W_{i^*}, v \in X\P_{j^*} }
[v \cdot A_{(w,v)}\U(t)].
\end{equation}

Summing Eq.~\ref{Eq:S_G_U_t} over all the cycles within the fill phase,

\begin{equation}\label{Eq:S_G_U}
S_{\CsU} \geq \frac{r}{M} \sum_{t \in P\mFILL} \sum_{ w \in X\W_{i^*}, v \in X\P_{j^*} }
[v \cdot A_{(w,v)}\U(t)]
\geq \frac{r}{M} O\mFILL_{\CsU}.
\end{equation}

In addition, by Proposition~\ref{prop:transmits_all_G_K}, $S_{\CsK} \geq O\mFILL_{\CsK}$. Therefore
\begin{equation}\label{Eq:S_G}
S_{\Cs} = S_{\CsK} + S_{\CsU}
\geq \frac{r}{M} (O\mFILL_{\CsK} + O\mFILL_{\CsU}) = \frac{r}{M} O\mFILL_{\Cs}.
\end{equation}
\end{proof}

We are now in a position to prove Theorem \ref{thm:SAMWP}.
\begin{proof}[Proof of Theorem \ref{thm:SAMWP}]

Every class $C_{(i,j)}$ is the selected class with probability $\frac{1}{\mW \cdot \mP}$. Using Lemma \ref{Cs_K_WP}
we therefore have for all
$i \in \set {1,2, \dots, \mW}$ and $j \in \set {1,2, \dots, \mP}$,
$S_{(i,j)} \geq \frac{r}{M \cdot \mW \cdot \mP} O\mFILL_{(i,j)}$.

Summing over all the classes, we obtain that the expected performance of our algorithm satisfies
\begin{equation} \label{eq:case0WV}
\sum_{i=1}^{\mW}
\sum_{j=1}^{\mP} S_{(i,j)} \geq
\frac{r}{M \cdot \mW \cdot \mP} \sum_{i=1}^{\mW}
\sum_{j=1}^{\mP}
O\mFILL_{(i,j)}.
\end{equation}

If \SAMWP\ is never \hfull\ during an arrival sequence, then $O_{(i,j)} = O\mFILL_{(i,j)}$ and therefore, by Eq.~\ref{eq:case0WV} the ratio between the performance of \opt\ and the expected throughput of \SAMWP\ is at most $\frac{M}{r} \cdot \mW \cdot \mP$, as required.

Assume next that \SAMWP\ becomes \hfull\ during an input sequence.
In such a case we compare the overall throughput due to packets {\em transmitted} by \SAMWP\ until the first cycle in which its buffer is empty again, and the profit obtained by \opt\ due to packets {\em accepted} by \opt\ during the same interval. We note that our analysis would also apply to subsequent such intervals, namely, until the subsequent cycle in which \SAMWP\ is empty again.

We note that in case \SAMWP\ becomes \hfull, \SAMWP\ holds in its buffer exactly $B$ $\Cs$-packets, and all these packets are transmitted by the time \SAMWP\ is empty again.
By the definition of $\deltaW$ in Section~\ref{sec:RC&S}, the maximal work which \SAMWP\ dedicates to any of these packets is at most $\deltaW$ times higher than the \emph {minimal} work which \opt\ dedicates to any $\Cs$-packet.
As a result, during the flush phase, in which \SAM\ handles $B$ $\Cs$-packets, \opt\ can handle at most
$\deltaW B + B$ $\Cs$-packets. Furthermore, by the definition of $\deltaP$ in Section~\ref{sec:RC&S}, the maximal profit of \opt\ from any $\Cs$-packet is at most $\deltaP$ higher than the \emph{minimal} profit of \SAMWP\ from any $\Cs$-packet. Combining the above reasoning implies that

\begin{equation}\label{Eq:delta_pi}
\frac{O\mFLUSH_{\Cs} }{S_{\Cs}} \leq \frac{\deltaW B + B}{B} \cdot \deltaP = (\deltaW + 1)\deltaP.
\end{equation}

As every class $C_{(i,j)}$ is the selected class w.p. $\frac{1}{\mW \cdot \mP}$, we have
$$\forall
i \in \set {1 \dots \mW},
j \in \set {1 \dots \mP},
S_{(i,j)} \geq
\frac{1}{(\deltaW + 1)\deltaP \cdot \mW \cdot \mP} O\mFLUSH_{(i,j)}.$$

Summing over all the classes we obtain
\begin{equation}\label{eq:case1WV}
\sum_{i=1}^{\mW}
\sum_{j=1}^{\mP}
S_{(i,j)} \geq
\frac{1}{(\deltaW + 1)\deltaP \cdot \mW \cdot \mP}
\sum_{i=1}^{\mW}
\sum_{j=1}^{\mP}
O\mFLUSH_{(i,j)}.
\end{equation}

Combining Equations~\ref{eq:case0WV} and~\ref{eq:case1WV} implies that the competitive ratio of \SAMWP\ is at most
\begin{equation}\label{Eq:case01WV}
\frac {
\sum_{i=1}^{\mW}
\sum_{j=1}^{\mP}
\left[ O\mFILL_{(i,j)} + O\mFLUSH_{(i,j)} \right]}
{
\sum_{i=1}^{\mW}
\sum_{j=1}^{\mP}
S_{(i,j)}}
\leq \left[ \frac{M}{r} + (\deltaW + 1)\deltaP \right] \cdot \mW \cdot \mP,
\end{equation}
which completes the proof.
\end{proof}

In the special case where all packets are known upon arrival, we obtain the following upper bound on the competitive ratio of \SAMWP:
\begin{corollary}
When $M=0$, \SAMWP\ is $O(\deltaW \cdot \deltaP \cdot \mW \cdot \mP)$-competitive.
\end{corollary}

\begin{proof}
We follow the proof of Theorem~\ref{thm:SAMWP}, and carefully check the required changes.

When all packets are known, Proposition~\ref{prop:transmits_all_G_K} remains essentially intact.
Furthermore, we have $S_{\Cs} = S_{\CsK} \geq O\mFILL_{\CsK} = O\mFILL_{\Cs}$, which replaces Lemma~\ref{Cs_K_WP}. Accordingly, Eq.~\ref{eq:case0WV} is modified to
\begin{equation}\label{eq:case0WV_M0}
\sum_{i=1}^{\mW}
\sum_{j=1}^{\mP} S_{(i,j)} \geq
\frac{1}{\mW \cdot \mP} \sum_{i=1}^{\mW}
\sum_{j=1}^{\mP}
O\mFILL_{(i,j)}.
\end{equation}

Equation~\ref{eq:case1WV} remains intact, as in deriving it we use the classify and randomly select scheme, independently of $M$. Combining Equations~\ref{eq:case0WV_M0} and~\ref{eq:case1WV} implies that when all packets are known, the competitive ratio of \SAMWP\ is at most
\begin{equation}\label{Eq:case01WV_M0}
\frac {
\sum_{i=1}^{\mW}
\sum_{j=1}^{\mP}
\left[O\mFILL_{(i,j)} + O\mFLUSH_{(i,j)}\right]}
{
\sum_{i=1}^{\mW}
\sum_{j=1}^{\mP}
S_{(i,j)}}
\leq \left[1 + (\deltaW + 1)\deltaP \right] \cdot \mW \cdot \mP,
\end{equation}
which completes the proof.
\end{proof}

\section{Proof Of Theorem \ref{thm:SAOWP}}\label{thm:SAOWP:proof}

\begin{proof}
\revchange{
We first consider the effect of uniformly at random selecting a class closure, instead of selecting a specific class.
First, note that the proof of Lemma \ref{Cs_K_WP} also directly applies to \SAOWP, implying that $S^*_{\Cs^*} \geq \frac{r}{M} O\mFILL_{\Cs}$.
Furthermore, the arguments used in the proof of Theorem~\ref{thm:SAMWP} also apply to \SAOWP, and in particular \SAOWP\ satisfies Equation~\ref{Eq:case01WV}, where we substitute in the denominator $S_{(i,j)}$ by $S^*_{(i,j)}$.
}

Consider next the affect of performing fill during flush.
In \SAOWP\ we accept packets also during the flush phase, but we never process any of these packets before all packets contributing to the algorithm being \hfull\ are transmitted, i.e., they are never processed before the flush phase is complete.
We enumerate the fill phases and the subsequent flush phases as follows: $P_{\FILL_1}, P_{\FLUSH_1}, P_{\FILL_2}, P_{\FLUSH_2}, \dots, P_{\FILL_n}, P_{\FLUSH_n}$, where $n \geq 1$. It should be noted that each such phase corresponds to a series of disjoint time intervals defined by the first cycle of the sequence of phases.
We further denote the $P_{\FLUSH_0}$ phase as an empty set of cycles, and in  case that the sequence ends by a fill phase, we also let $P_{\FLUSH_n}$ denote an empty set of cycles.
\revchange{
Similarly, we further define $P^*_{\FILL_i}, P^*_{\FLUSH_i}$, for the appropriate values of $i$, to denote the fill and flush phases corresponding to \SAOWP.
}

Denote the profit accrued by \SAMWP\ and \opt\ from packets which arrive during the $i^{th}$ fill phase by $S^{(P_{\FILL_i})}$ and $O^{(P_{\FILL_i})}$ respectively. Similarly, denote the profit of \SAMWP\ and \opt\ obtained from packets which arrive during the $i^{th}$ flush phase by $S^{(P_{\FLUSH_i})}$ and $O^{(P_{\FLUSH_i})}$, respectively.
Similarly, we let $S^{* (P^*_{\FILL_i})}$ and $S^{* (P^*_{\FLUSH_i})}$ indicate the profit of \SAOWP\ obtained from packets which arrive during its $i^{th}$ fill and flush phase, respectively.

Using this notation, we recall that, by the analysis of \SAMWP\ presented in Theorem~\ref{thm:SAMWP}
\begin{equation} \label{eq:P_old}
O^{(P_{\FILL_i})} + O^{(P_{\FLUSH_i})} \leq \left[ \frac{M}{r} + (\deltaW + 1) \deltaP \right]
\mW  \cdot \mP \cdot S^{(P_{\FILL_i})}
\end{equation}
for every $i=1,\ldots,n$.

This induces an implicit mapping $\phi$ of
the units of profit obtained from
$\Cs$-packets accepted by \opt\ during $P_{\FILL_i} \cup P_{\FLUSH_i}$ to the units of profit obtained from $\Cs$-packets accepted by \SAMWP\ during $P_{\FILL_i}$ (either known, or unknown that were parsed), such that every unit of profit obtained by \SAMWP\ has at most $\left[\frac{M}{r} + (\deltaW + 1) \deltaP \right] \mW  \cdot \mP$ units of profit mapped to it.

A key observation is noting that the image of mapping $\phi$ is essentially the profit attained from the set of $\Cs$-packets contributing to the algorithm being \hfull\ at the end of the corresponding fill phase.

As \SAOWP\ may accept packets during flush, in the beginning of the subsequent fill phase the buffer of \SAOWP\ may not be empty. In particular, there could be $\Cs$-packets accepted during the recent flush phase that are stored in the buffer.
However, none of these packets have any \opt\ packets mapped to them.
It follows that these packets can contribute to \SAOWP\ becoming \hfull\ in the new fill phase, and any profit implicitly mapped to the profit of these packets by $\phi$ would correspond to packets arriving during the new fill phase, or its subsequent
flush phase.
The implicit mapping is depicted in Fig.~\ref{fig:fill_during_flush}, along with the difference between the mapping arising from the behavior of \SAMWP\ (visualized above the time axis), and the mapping arising from the behavior of \SAOWP (visualized below the time axis). Note that the fill and flush phases of both algorithms need not be synchronized, since \SAOWP\ can potentially become \hfull\ ``faster'' than \SAMWP.

\begin{figure*}[t]
\centering
\includegraphics[width=1.0\textwidth]{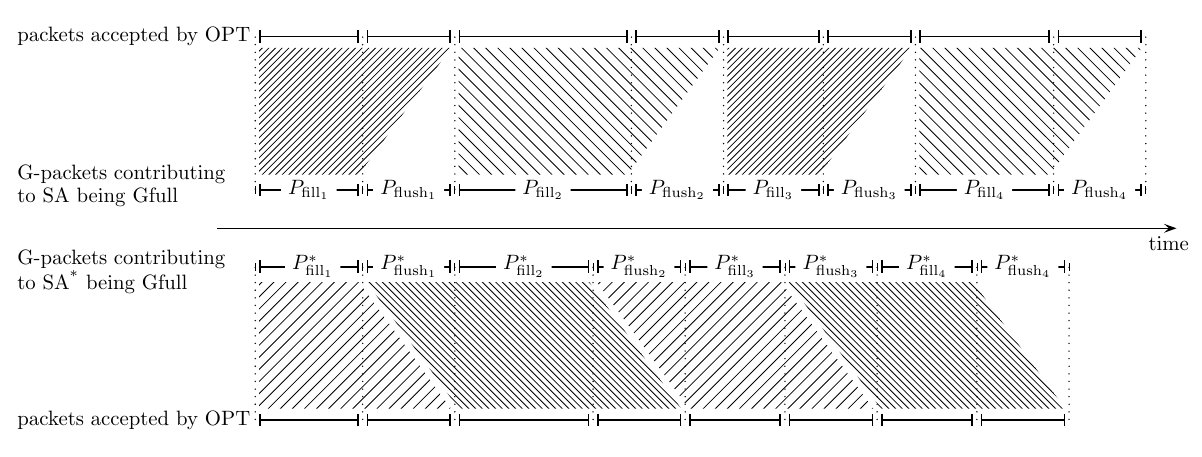}
\caption{Visualization of the mappings induced by the analysis of \SAMWP\ and \SAOWP, for the first 4 fill and flush phases. The fill and flush phases of \SAMWP\ are denoted $P_{\FILL_i}$ and $P_{\FLUSH_i}$, respectively, whereas the fill and flush phases of \SAOWP\ are denoted $P^*_{\FILL_i}$ and $P^*_{\FLUSH_i}$, respectively.
The top part shows the mapping of profit corresponding to packets accepted by \opt\ along time, to the profit corresponding to $\Cs$-packets accepted by \SAMWP\ during the fill phase (since \SAMWP\ does not accept any packets during the flush phase).
The bottom part shows the induced mapping of profit obtained by packets accepted by \opt\ along time to the profit of $\Cs$-packets accepted by \SAOWP\ during both the preceding flush phase, and the current fill phase.}
\label{fig:fill_during_flush}
\end{figure*}

It follows that Eq.~\ref{eq:P_old} now translates to
\begin{equation} \label{eq:P_new}
\begin{array}{l}
O^{(P^*_{\FILL_i})} + O^{(P^*_{\FLUSH_i})} \leq \\
\quad \quad \quad \left[ \frac{M}{r} + (\deltaW + 1) \deltaP \right] \mW  \cdot \mP \cdot \left[S^{* (P^* _{\FLUSH_{i-1}})} + S^{* (P^*_{\FILL_i})}\right]
\end{array}
\end{equation}
for every $i=1,\ldots,n$.
Summing over all $i=1,\ldots,n$, we obtain that the competitive ratio guarantee for \SAOWP\ is the same as that for \SAMWP.

Lastly, the analysis of \SAMWP\ does not assume any specific scheduling rule to be applied, as long as the $\Cs\K$-first order rule is maintained.
Thus, our competitive ratio guarantee is independent of the specific ordering within the set of $\CsK$-packets, as well as within the set of non-$\CsK$-packets.
\end{proof}

\section{Running Example of \SAM}\label{App:run_Example}
\revchange{
Figure~\ref{fig:algorithm-example} exemplifies a running of \SAM\ equipped with a 3-slots buffer. Each packet is represented by a square. If it is a known $(w, v)$-packet, then $(w,v)$ (namely, its work, profit values, resp.) appears within the square representing the packet. If the packet is unknown, however, the (unknown) work and profit values do not appear, and the packet's color is dark gray.

Known packets which belong to the selected class ($\CsK$-packets) are marked in light gray.
The figure assumes that the (randomly-chosen) selected class is the class of packets with work- and profit- values within the range [3, 4]. Recall, that this range refers to the
characteristics of a packet \emph{upon arrival}.
For instance, a $(3, 3)$-packet always belongs to the selected class, although after being processed its residual work decreases, and it becomes a $(2, 3)$-packet, and later a $(1, 3)$-packet, and so on.

Each cycle begins with the transmission step, in which a fully processed packet, if such exists, leaves the queue.
In our example there is no packet transmitted since we focus our attention on handling arrivals and determining priorities which are the core components of our algorithm.
This step is followed by the arrival step,
where arriving packets are handled by the algorithm. Finally,
the cycle ends with a processing step, where the head-of-line (HoL) packet is processed. This packet is emphasized by an extra internal square.
The state of the queue at the end of each cycle is depicted by a light-gray background.
At each cycle the algorithm tosses a coin, and assigns the cycle as an \emph{admittance cycle} w.p. $r$. In this example, we assume that cycles $1,3,5, 6$ are admittance cycles.
We now turn to explain the scenario depicted in Figure~\ref{fig:algorithm-example} cycle by cycle.
\paragraph{\textbf{$t=0$}} Begin with an empty buffer.
\paragraph{\textbf{$t=1$}} A known $(4,4)$-packet arrives. As both its work- and profit- values belong to the ranges [3,4], it is a $\CsK$-packet, and therefore it is retained by the algorithm (recall that $\CsK$-packets are never dropped during the fill phase, as shown in Proposition~\ref{prop:transmits_all_G_K}).

Next, a $U$-packet arrives. As this is an admittance cycle, this $U$-packet is \emph{admitted}, that is, accepted into the buffer, and assigned to the HoL.
Since this is the last packet to arrive in this cycle, and being the HoL-packet, this packet is processed in the processing step. We refer to this packet as being \emph{parsed} in this cycle, as this is the first processing cycle of this packet.

After parsing, the characteristics of
the HoL packet become known: it is now a known $(1,8)$-packet.
Namely, when it arrived, it was a $(2, 8)$-packet which has received one cycle of processing.
By these values, this packet does not belong to the selected class. Therefore, it is pushed down to the buffer's tail. Instead, the $\CsK$-packet, with values $(4,4)$ is assigned to be the HoL packet.
It should be noted that although the parsed $(1,8)$-packet is superior to any $\CsK$-packet currently in the buffer (since it carries a profit value of 8 while requiring just one more cycle of work), \SAM\ still prefers $\CsK$-packets over this packet. We note that the improved \SAO\ algorithm would re-assign such a packet to be a $\CsK$-packet by considering the selective class closure.

\paragraph{\textbf{$t=2$}} No packets arrive. The HoL-packet, $(4,4)$, is processed, and becomes a $(3,4)$-packet.

\paragraph{\textbf{$t=3$}}
This is an admittance cycle. Therefore, the first arriving $U$-packet is admitted. In particular, this cycle well exemplifies the buffer's ordering: at top-priority is the admitted packet; at a second priority is the $\CsK$-packet, $(3,4)$; the remaining packet in the buffer, $(1,8)$, is of a lowest priority.

When a second $U$-packet arrives, \SAM\ tosses a coin, and replaces the previously-admitted packet with the new arriving $U$-packet w.p. $1/2$. When a third $U$-packet arrives, \SAM\ tosses a coin again, and replaces the previously-admitted packet with the new arriving $U$-packet w.p. $1/3$.

In the processing step, \SAM\ parses the admitted packet, unraveling it as a $(3,3)$-packet. Namely, upon arrival its characteristics were $(4,3)$, ascribing it to the selected class. As there already exists another $\CsK$-packet in the buffer (the $(3,4)$-packet) \SAM\ breaks the tie between the two $\CsK$-packets in its buffer by FIFO order.
We note that the improved \SAO\ algorithm would transition to the flush phase at this point, since it would have been full of $\CsK$-packets.

\paragraph{\textbf{$t=4$}} First, we have an arriving known $(2,5)$-packet. By its characteristics, it is not a $\CsK$-packet. Therefore, it is assigned the lowest priority. In particular, as the buffer is full, this packet is discarded. Next, a $U$-packet arrives. However, as this is a non-admittance cycle, the $U$-packet is discarded as well. Finally, during the processing step, the HoL packet is processed, decreasing its remaining work to 2.

\paragraph{\textbf{$t=5$}} We have a single arriving $U$-packet. As it is an admitted cycle, this $U$-packet is admitted, hence, accepted and parsed. In order to make room for this admitted packet, the $(1,8)$-packet in the tail is pushed-out and dropped.
 After parsing, the $U$-packet is uncovered as a $(1, 2)$-packet. Namely, upon arrival it was a $(2,2)$-packet. By these characteristics, this packet does not belong to the selected class, and therefore has the lowest priority, and downgraded to the tail.

\paragraph{\textbf{$t=6$}} This is an admittance cycle. Therefore the first arriving $U$-packet is admitted, pushing-out from the buffer the $(1,2)$-packet, which was in the tail. When a second $U$-packet arrives, it replaces the previously-admitted packet w.p. $  1/2$. Then, a $(2,7)$-packet arrives. By its characteristics, it is neither an admitted packet (as it is a $K$-packet), nor does it belong to the selected class. As a result, the $(2,7)$-packet is assigned the lowest priority, and is therefore discarded.
The last arrival in this cycle is a known $(4,4)$-packet. By its characteristics, it is a $\CsK$-packet. Since the buffer already contains $B-1=2$ $\CsK$-packets, the $U$-packet at the HoL is dropped, and the newly arriving $CsK$-packet is accepted to the queue (see lines~\ref{alg:SAM:if_is_Gk}-\ref{alg:SAM:accept_Gk} in Algorithm~\ref{alg:SAM}).
The queue therefore becomes \hfull, i.e., the buffer is full with $\CsK$-packets. \SAM\ then switches to the \emph{flush} state, and it will merely process all the packets in its buffer in a run-to-completion manner and transmit all the fully-processed packets, until the buffer is empty again.

\begin{figure}
    \centering
    \begin{subfigure}[b]{1.0\textwidth}
        \includegraphics[width=\textwidth]{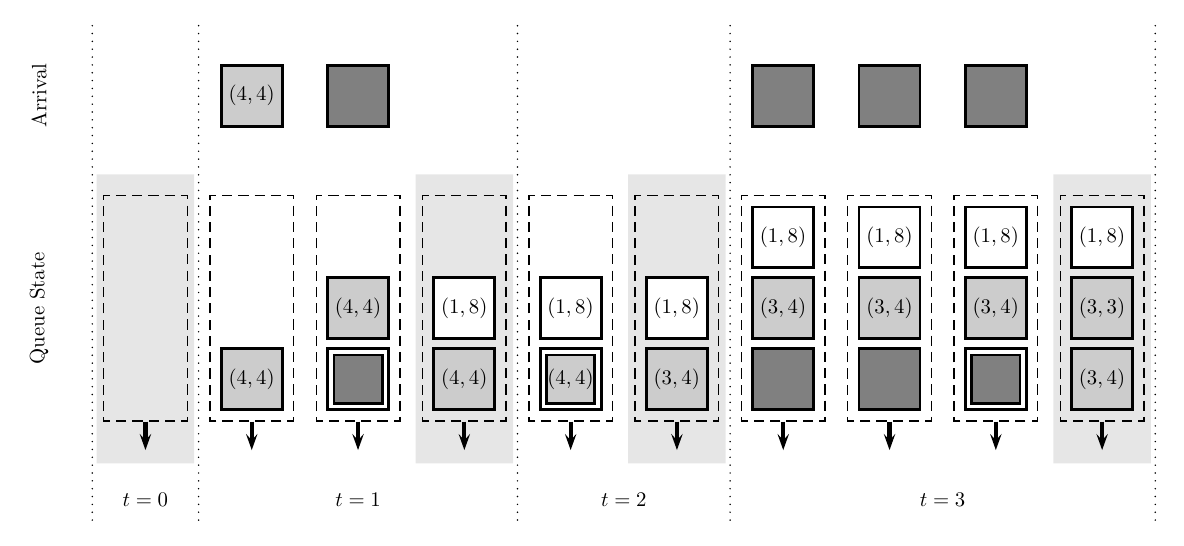}
        \label{fig:algorithm-example_part1}
    \end{subfigure}
    \begin{subfigure}[b]{1.0\textwidth}
        \includegraphics[width=\textwidth]{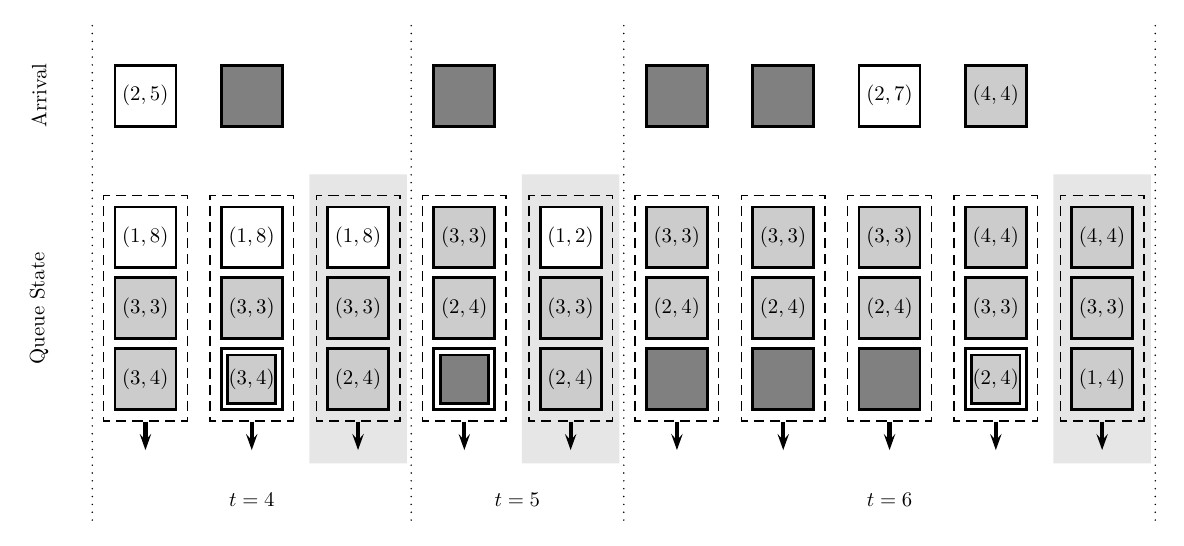}
        \label{fig:algorithm-example_part1}
    \end{subfigure}
    \caption{Running Example of \SAM, equipped with a 3-slots buffer. Each known packet is labeled $(w(p),v(p))$, where $w$ is the remaining work and $v$ is the profit. $\CsK$-packets are marked by light gray. $U$-packets are colored by dark gray. \\
    Each cycle begins with a transmission step, in which a fully-processed packet, if such exists, is transmitted. Next comes the arrival step, where arriving packets are handled by the algorithm one by one. For each arriving packet, the buffer below the arrival depicts the state of the buffer after handling the packet's arrival.
 The packet in the queue's head-of-line (HoL) at the end of the arrival step is emphasized by an extra, internal, square. This packet is the one processed in the processing step. The state of the buffer at the end of each cycle is highlighted with light-gray background.}

    \label{fig:algorithm-example}
\end{figure}
}

\end{document}